\documentclass[algo2e, 10pt,twocolumn]{IEEEtran}
\usepackage{soul,xcolor,lipsum,enumitem}
\usepackage[linesnumbered,ruled,vlined]{algorithm2e}

\SetCommentSty{mycommfont}
\usepackage{algorithmicx,comment,verbatim}
\usepackage{tablefootnote}
\usepackage{algpseudocode}
\usepackage{amsfonts}
\usepackage{amsmath,empheq}
\usepackage{booktabs}
\usepackage{upgreek}
\usepackage{array}
\usepackage{amssymb}
\usepackage{amsthm}
\usepackage[ansinew]{inputenc} 
\usepackage[colorlinks=true, allcolors=blue]{hyperref}
\usepackage{dsfont}
\usepackage{mathtools}
\usepackage{graphicx}
\usepackage[skip=1pt,font=scriptsize]{subcaption}
\usepackage{import}
\usepackage{multirow}
\usepackage{cite}
\usepackage[export]{adjustbox}
\usepackage{breqn}
\usepackage{mathrsfs}
\usepackage{acronym}
\usepackage{balance}
\usepackage{setspace}
\usepackage{stackengine}
\usepackage{steinmetz}
\usepackage[skip=2pt,font=footnotesize]{caption}

\newcommand{\nn}{\nonumber}
\newcommand{\mc}{\mathrm c}
\newcommand{\mg}{\mathrm g}
\newcommand{\sigc}{\sigma_{\mathrm c}}
\newcommand{\sigg}{\sigma_{\mathrm g}}
\newcommand{\M}{\Phi}

\newcommand{\etas}{\eta_{(s)}}
\newcommand{\deltas}{\delta_{(s)}}
\newcommand{\gammas}{\gamma_{(s)}}
\newcommand{\etag}{\eta_{\mathrm{g}}}
\newcommand{\etad}{\eta_{\mathrm{dgd}}}
\newcommand{\bg}{b_{\mathrm{g}}}
\newcommand{\bd}{b_{\mathrm{dgd}}}
\newcommand{\Bg}{B_{\mathrm{g}}}
\newcommand{\Bd}{B_{\mathrm{dgd}}}

\newcommand{\x}{\mathbf x}
\newcommand{\DNR}{\mathrm{DNR}}

\newcommand{\GCR}{\mathrm{GCR}}
\newcommand{\xstar}{\mathbf x^\star}
\newcommand{\xloc}{\mathbf x^{\mathrm{loc}}}
\newcommand{\Eloc}{\mathrm{E}^{\mathrm{loc}}}
\newcommand{\W}{\mathbf W}

\newcommand{\I}{\mathbf I}
\newcommand{\err}{\varepsilon}
\newcommand{\errt}{\varepsilon_{\mathrm{th}}}
\newcommand{\Es}{\overline{\mathrm E}_{(s)}}
\newcommand{\E}[1]{\overline{\mathrm E}_{(#1)}}
\newcommand{\Esn}{\overline{\mathrm E}_{(s+1)}}
\newcommand{\init}{{\mathrm{init}}}
\newcommand{\gnd}{{\mathrm{gnd}}}
\newcommand{\dnr}{{\mathrm{dnr}}}
\newcommand{\cnd}{{\mathrm{cnd}}}

\makeatletter

\def\@seccntformat#1{\csname the#1\endcsname\quad}
\newcommand{\twocol}[1]{%
\if@twocolumn
#1
\fi}
\newcommand{\onecol}[1]{%
\if@twocolumn
\else
#1
\fi}
\newcommand{\onetwocol}[2]{%
\if@twocolumn
#2
\else
#1
\fi}
\makeatother

\newtheorem{theorem}{Theorem}
\newtheorem{assumption}{Assumption}
\newtheorem{lemma}{Lemma}
\theoremstyle{definition}

\newtheorem{corollary}{Corollary}
\newtheorem{remark}{Remark}
\theoremstyle{remark}

\usepackage{cancel}
\usepackage{bbm}

\graphicspath{{./figures/}}
\newcommand{\secref}[1]{Sec.~\ref{#1}}

\newcommand{\1}{\mathbf 1}

\newcommand\nm[1]{\textcolor{blue}{\bf NM:[ #1]}}

\newcommand\mainfig{\begin{figure*}[t]
\centering

\begin{subfigure}{0.48\linewidth}
    \centering
    \includegraphics[width=.9\linewidth]{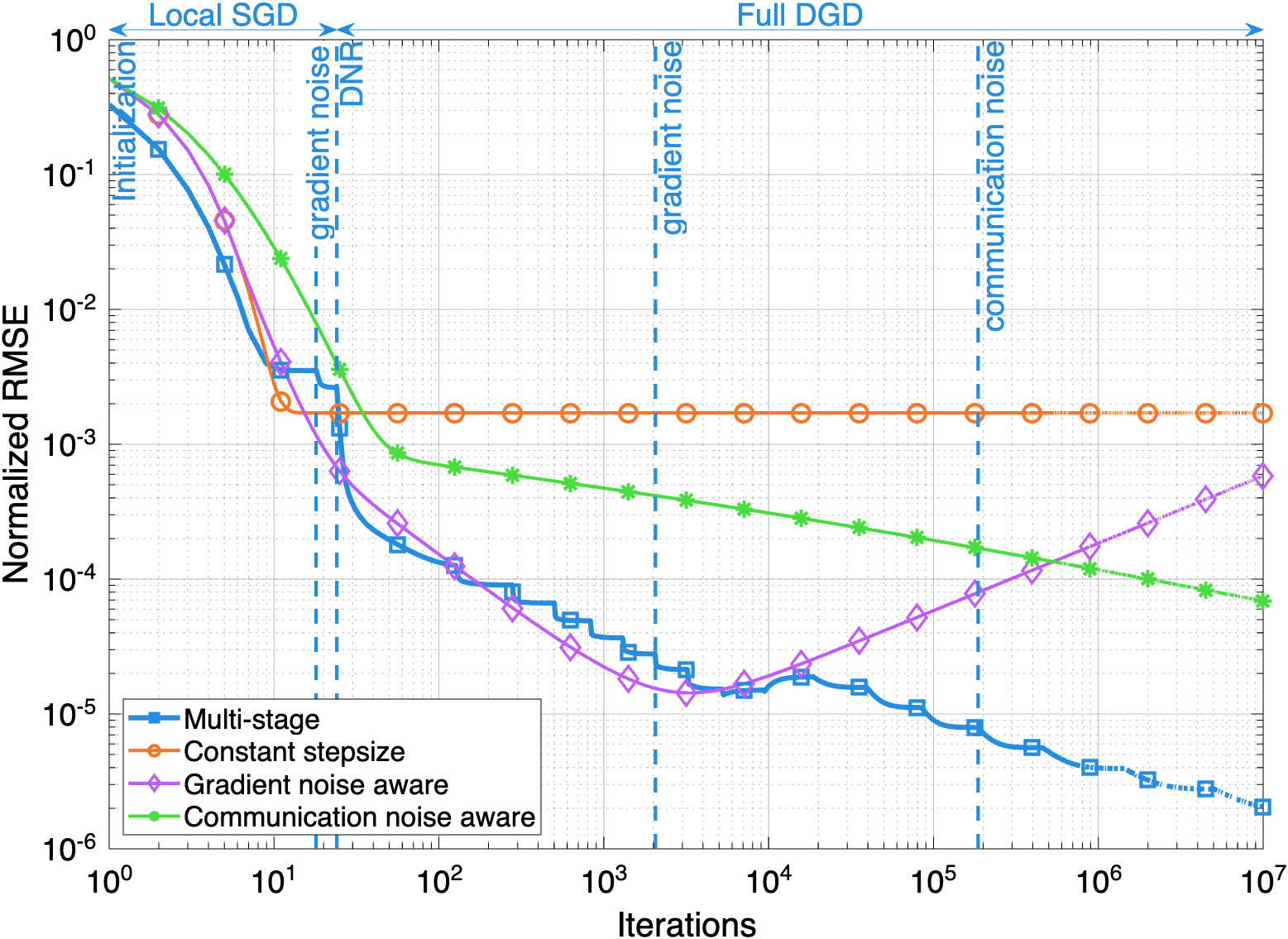}
    \label{fig:comp_set4}
\end{subfigure}
\hfill
\begin{subfigure}{0.48\linewidth}
    \centering
    \includegraphics[width=.9\linewidth]{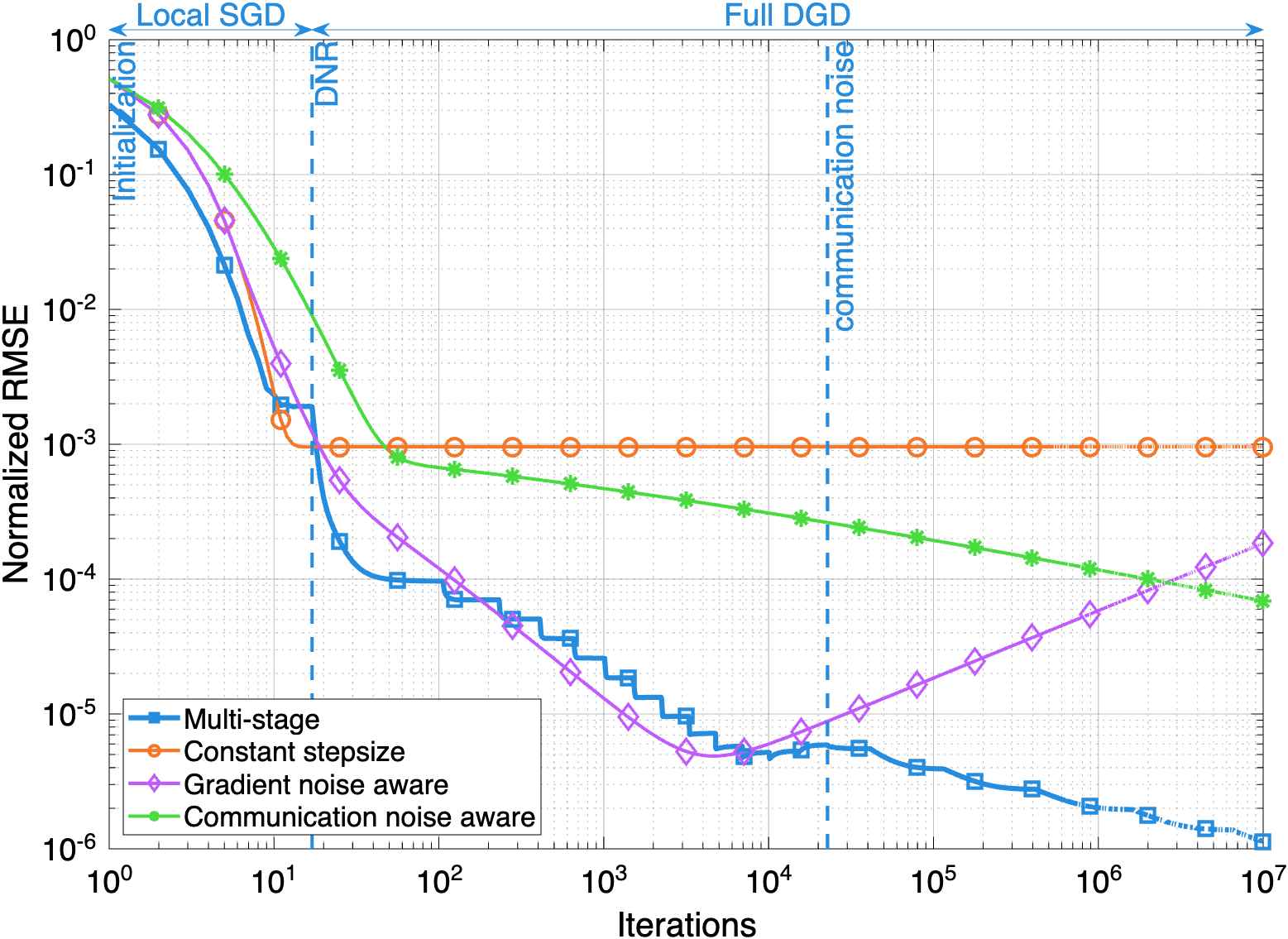}
    \label{fig:comp_set2}
\end{subfigure}

\vspace{-2mm}

\begin{subfigure}{0.48\linewidth}
    \centering
    \includegraphics[width=.9\linewidth]{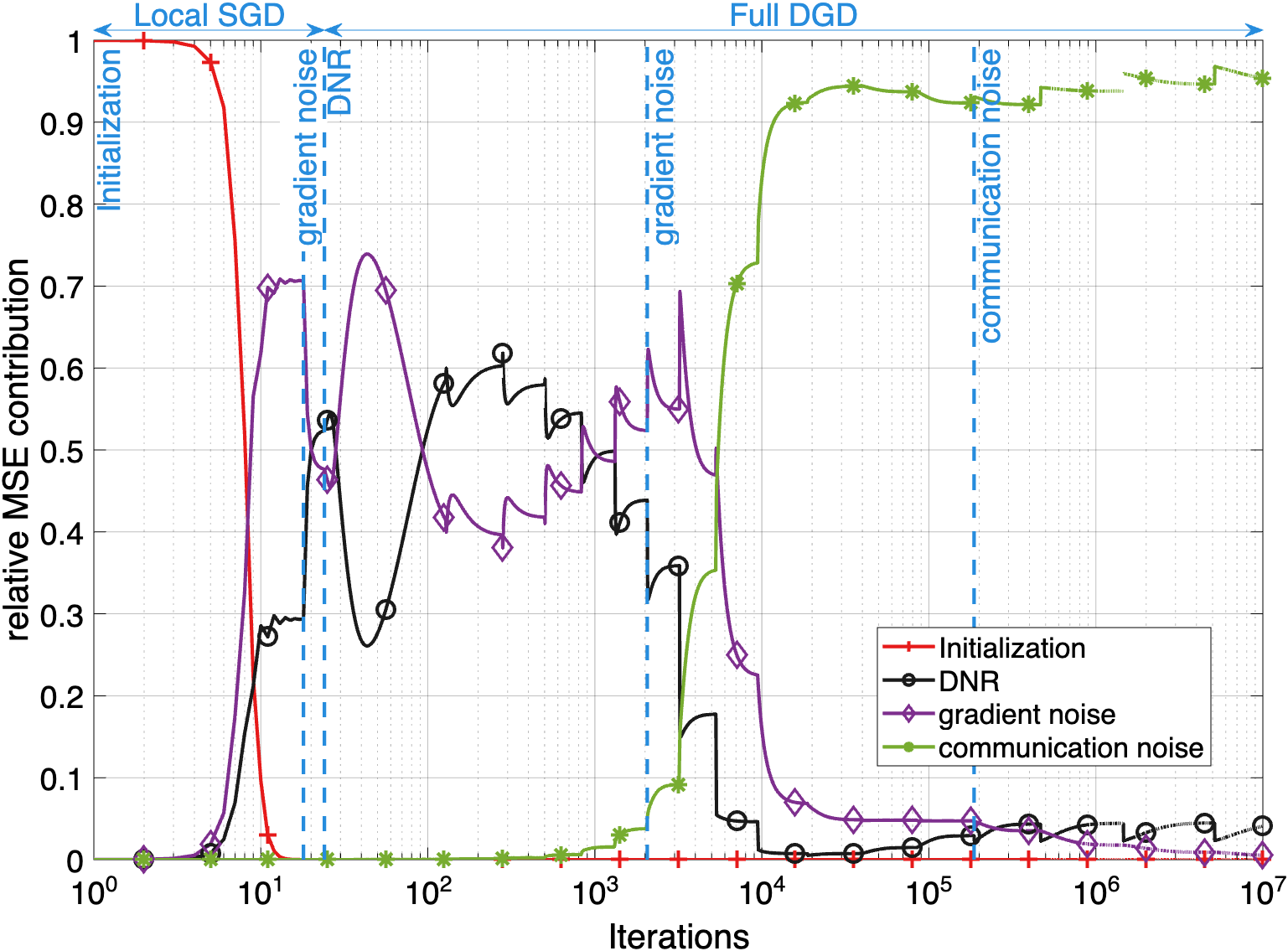}
    \caption{High-GCR, intermediate DNR-to-gradient-noise ratio setting: \(\DNR{\approx}1.2{\cdot} 10^6\), \(\GCR{=}  3.7{\cdot}10^6\), \onecol{\\}\(\sigc{\approx} 3{\cdot} 10^{-4}\|\nabla f(\xstar)\|\), \(\sigg{\approx} 1.4\|\nabla f(\xstar)\|\).}
    \label{fig:perc_set4}
\end{subfigure}
\hfill
\begin{subfigure}{0.48\linewidth}
    \centering
    \includegraphics[width=.9\linewidth]{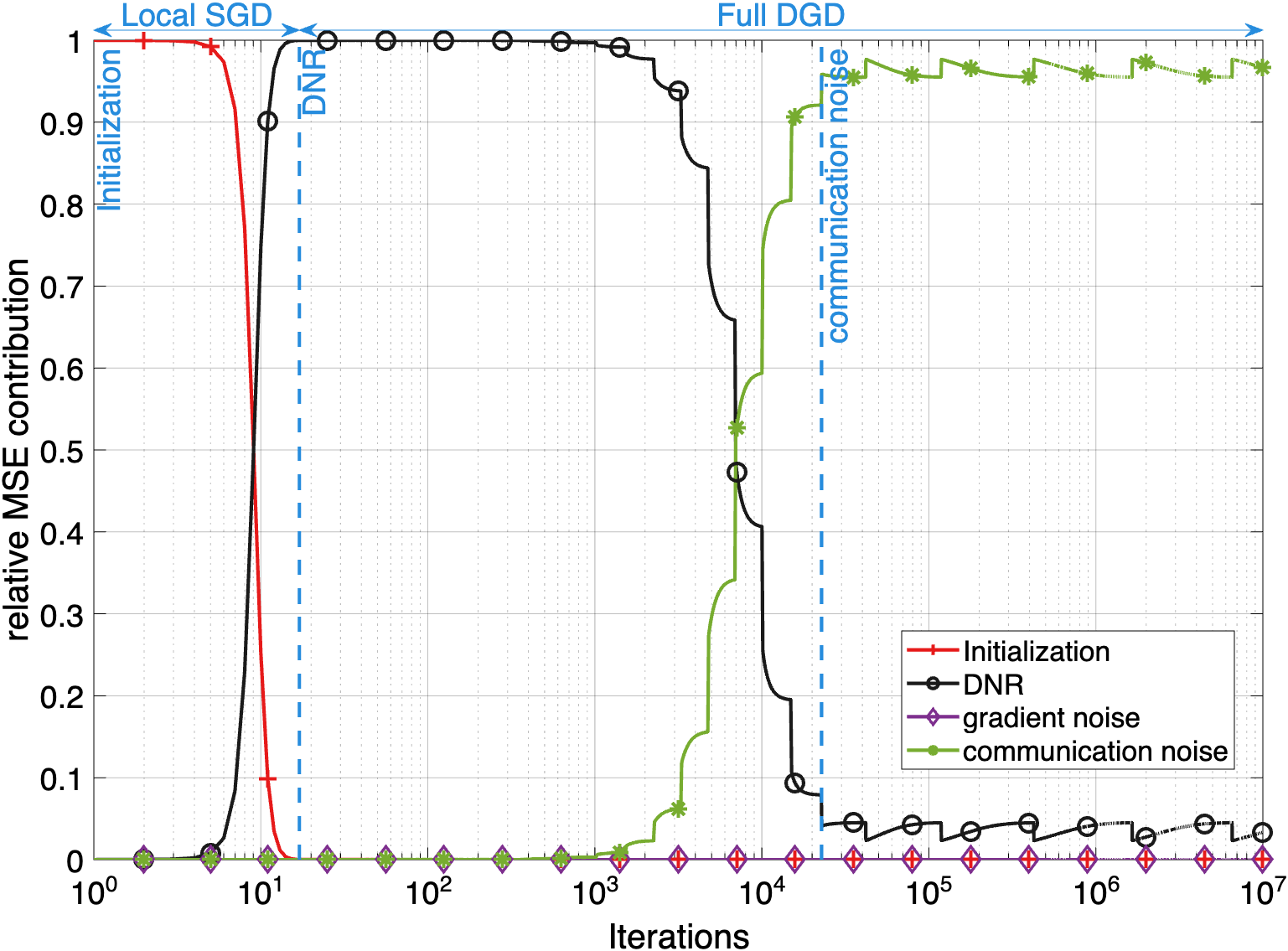}
    \caption{High DNR-to-gradient-noise ratio setting:\onecol{\\} \(\DNR{\approx}1.2\cdot 10^6\), \(\GCR{=}1\),
    \onecol{\\}
     \(\sigc{\approx} 
    9.1{\cdot} 10^{-5}\|\nabla f(\xstar)\|\), \(\sigg{\approx} 2.3{\cdot}10^{-4}\|\nabla f(\xstar)\|\).}
    \label{fig:perc_set2}
\end{subfigure}
\vspace{3mm}
\caption{
Top row: normalized RMSE under different stepsize schedules.
Bottom row: relative contributions of initialization, DNR, gradient noise, and communication noise to the overall MSE. The plots show also the phase and regime transitions under the proposed multi-stage schedule.
}
\twocol{\vspace{-3mm}}
\label{fig:set4}
\end{figure*}
}

\title{Decentralized Gradient Descent: \\Bottleneck Regimes and Budget Complexity}

\author{Nicol\`o Michelusi,~\IEEEmembership{Senior Member,~IEEE}
\thanks{N. Michelusi is with the School of Electrical, Computer and Energy Engineering, Arizona State University.}
\thanks{This research has been funded in part by NSF under grant CNS-2129015.}
}

\begin{document}

\setstcolor{red}
\setulcolor{red}
\setul{red}{2pt}
\maketitle

\onecol{\vspace{-15mm}}

\begin{abstract}
Decentralized gradient descent (DGD) is widely used for solving distributed optimization problems over networks of agents. While its convergence properties are well understood, less is known about the communication and computation resources required to attain a prescribed accuracy.
In this paper, we study DGD from a resource-aware perspective and characterize the communication-computation budget required to attain a target error level. We develop a bottleneck-centric framework in which different factors dominate the optimization dynamics at different error scales. Specifically, we identify operating regimes governed by initialization, objective heterogeneity and network connectivity, gradient noise, and communication noise. To capture these effects, we introduce two fundamental quantities: the gradient-Diversity-to-Network-connectivity Ratio (DNR) and the Gradient-to-Communication-noise Ratio (GCR). We show that these quantities determine the sequence of bottlenecks encountered during optimization and the corresponding budget-optimal operating strategy. Using a multi-stage analysis, we derive optimal stepsize selections and explicit budget-complexity bounds that quantify the budget resources required to attain a prescribed accuracy. The resulting expressions reveal how the overall budget decomposes into contributions associated with successive bottlenecks and provide insight into the fundamental tradeoffs among objective heterogeneity, network connectivity, gradient noise, and communication noise.
\end{abstract}
\onecol{\vspace{-5mm}}
\section{Introduction}

Many emerging systems consist of large collections of interconnected devices that must collaboratively solve optimization problems using locally available data. Examples arise in machine learning, distributed inference, estimation, remote sensing, and multi-agent control, where the objective is to learn a common model or estimate a global parameter from information dispersed across the network \cite{6494683,Nedic2018,YANG2019278}. In such settings, \(N\) devices seek to cooperatively solve
\begin{align}
\label{global}
x^\star= \arg\min_{x \in \mathbb{R}^d}\ \frac{1}{N} \sum_{i=1}^N f_i(x),
\tag*{\{P\}}  
\end{align}
where \(f_i:\mathbb{R}^d\mapsto\mathbb{R}\) is a local objective function known only to node \(i\). In machine-learning applications, \(f_i\) typically represents the empirical loss associated with a local dataset, whereas in distributed estimation problems it quantifies the mismatch between local measurements and the underlying parameter of interest \cite{10103556}.
The challenge is to solve \ref{global} efficiently despite distributed information and limited communication resources.

In many applications, including swarms of uncrewed aerial vehicles operating in remote areas \cite{9475989}, solving \ref{global} without centralized data aggregation is paramount \cite{8950073}.
Decentralized optimization algorithms, such as decentralized gradient descent (DGD)  \cite{Nedic2009grad,Yuan2016}, 
solve \ref{global} by alternating between local computation and information exchange with neighboring nodes. 
Over the past decade, a large body of work has established convergence guarantees for DGD and its variants under a variety of assumptions on the objective functions and network topology \cite{Nedic2009,Yuan2016}. 
Subsequent works have investigated stochastic gradients \cite{Lian17}, communication impairments 
 \cite{Kar2009,10947567,10680589,michelusi26,9562482,9772390,9716792,9838891}, as well as communication-efficient schemes based on quantization and compression \cite{8786146,Koloskova2019,9782148}. 
 Collectively, these works have revealed the key role of objective heterogeneity, network connectivity, gradient and communication noises.

While convergence rates are now relatively well understood,  less is known about the communication and computation resources required to attain a prescribed accuracy. This question is becoming increasingly important in large-scale and resource-constrained systems, where communication latency, transmission energy, and computational costs may dominate the overall optimization process. A key observation is that reducing the optimization error is not equally difficult across all stages of the algorithm: when the error is large, local computation may suffice, whereas at higher accuracy levels other factors, such as objective heterogeneity, gradient noise, or communication noise, may become dominant.
This suggests that the optimization process should be viewed as a sequence of stages, each governed by a different dominant factor and therefore requiring a different operating strategy.

This observation motivates the bottleneck-centric perspective adopted in this paper. 
 Rather than treating all error sources as simultaneously active throughout the optimization process, we identify the bottleneck governing convergence at each error scale and characterize the most cost-effective strategy for overcoming it.
 This viewpoint leads to a multi-stage formulation in which different bottlenecks become active at different stages of the optimization trajectory, giving rise to distinct operating regimes and fundamentally different budget scalings.

 
 Since the resulting budget-allocation problem is generally intractable, we instead seek a tractable characterization of the dominant cost components that determine the budget required to attain a prescribed accuracy. To this end, our analysis identifies two fundamental quantities governing the optimization dynamics: the gradient-Diversity-to-Network-connectivity Ratio (DNR) and the Gradient-to-Communication-noise Ratio (GCR). The former captures the tension between objective heterogeneity and network connectivity, whereas the latter captures the relative importance of gradient and communication noises. Together, these quantities determine which bottleneck is active at a given error scale and therefore dictate the budget-optimal operating mode and parameter selection.
 Building on this characterization, we derive explicit budget-complexity bounds that quantify the dependence of the required budget on the DNR, GCR, target accuracy, and noise levels, and reveal the contributions of successive bottlenecks. 
 Figure~\ref{fig:bottleneck} schematically summarizes the main budget-complexity characterization developed in this paper.
 
 More broadly, the proposed bottleneck-centric perspective opens a new direction for adaptive decentralized optimization frameworks.

\onecol{\vspace{-5mm}}
\subsection{Related Work}

Early works established convergence bounds for DGD under strongly convex and smooth objectives over connected networks, revealing the impact of objective heterogeneity and network topology on performance \cite{Nedic2009grad,Yuan2016}. Unlike centralized gradient descent, which converges linearly to the global optimum under strong convexity, DGD with a constant learning stepsize converges linearly only to a neighborhood of the optimum \cite{Yuan2016,10947567}. This bias stems from the combination of objective heterogeneity and decentralized information exchange.
To eliminate this bias, several works have considered decreasing learning stepsizes, under which the distance to the global optimum converges at a sublinear rate of \(\mathcal O(1/t)\) \cite{10947567}. \emph{These results reveal that, even in the absence of stochastic perturbations, objective heterogeneity and limited network connectivity fundamentally limit the convergence behavior of DGD.}

Several exact decentralized optimization methods, including EXTRA, exact diffusion, and gradient tracking \cite{Shi2015EXTRA,8491372,7798263,7398129}, have been proposed to eliminate the heterogeneity-induced bias of DGD. These methods may improve convergence performance but require additional communication, memory, and computational resources.


Several works have extended DGD to stochastic settings in which the exact local gradients are replaced by estimates obtained from minibatch sampling or noisy observations \cite{10947567,Lian17,Tang2018}.
In terms of the root mean squared error (RMSE) between the iterates and the global optimum, gradient noise slows the convergence rate of DGD from \(\mathcal O(1/t)\) to \(\mathcal O(1/\sqrt{t})\) in the strongly convex setting \cite{7405263}. These results suggest that \emph{gradient noise may become the dominant limitation once the error is sufficiently small, thereby inducing a bottleneck distinct from that associated with objective heterogeneity and network connectivity.}

Another line of work has investigated decentralized optimization under imperfect or constrained communications, including link failures, channel noise, fading, interference, unreliable links, rate-constrained networks, and quantized communications \cite{Kar2009,10680589,michelusi26,9562482,9716792,8786146}. 
Unlike gradient noise, communication noise directly affects the consensus process itself, creating a tradeoff between information exchange and noise accumulation,
typically controlled via a consensus stepsize\cite{8786146,10680589,10947567,9563232}. As a result, communication noise leads to a more severe degradation of the convergence behavior, yielding an RMSE decay of \(\mathcal O(1/\sqrt[4]{t})\) \cite{8786146,10947567}. This suggests that \emph{communication noise may eventually become the dominant bottleneck limiting further error reduction.}

Communication-efficient decentralized optimization has also been studied through quantization, compression, and related techniques \cite{Kovalev2020,9782148,Liao2021,Magnusson2020}. While these methods reduce communication overhead and can preserve favorable convergence guarantees, they primarily focus on algorithm design rather than on characterizing the budget implications of different bottlenecks.


The communication-computation tradeoff in decentralized optimization has also been recently investigated through multi-consensus and nested DGD methods \cite{berahas2018balancing,berahas2021convergence,choi2022convergence}. These approaches perform multiple consensus and/or gradient updates within each iteration  to improve  convergence. Their analyses demonstrate that the balance between communication and computation can significantly affect the efficiency of decentralized optimization and that increasing the amount of communication may substantially accelerate convergence. However, these works focus on the design and analysis of specific algorithmic architectures and do not explicitly characterize how the communication-computation budget should evolve as different impairments become active throughout the optimization process.

Taken together, the above works suggest that different impairments may dominate the optimization dynamics at different error scales. However, existing analyses largely focus on a specific impairment or algorithmic enhancement. In contrast, we develop a bottleneck-centric framework that characterizes how different bottlenecks become active throughout the optimization process and how they jointly determine the budget required to attain a prescribed accuracy.


\begin{figure*}[t]
\centering
\includegraphics[width=\onecol{.8}\textwidth]{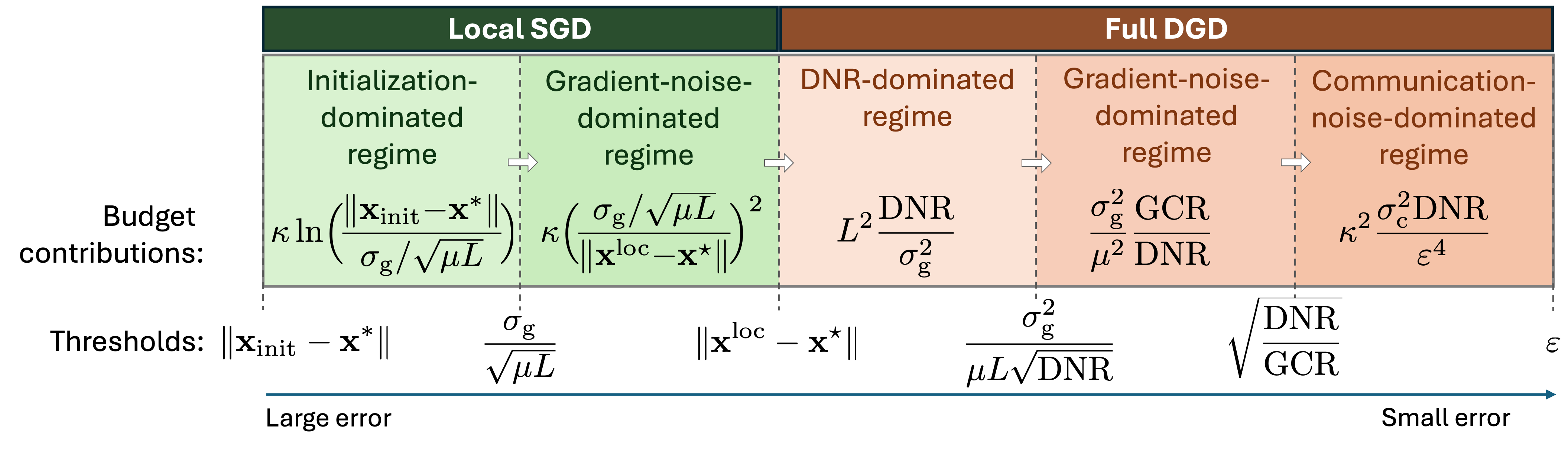}
\twocol{\vspace{-6mm}}
\caption{
Representative optimization trajectory under the proposed bottleneck-centric framework (high-GCR, intermediate-DNR setting studied in \secref{highGCRmedDNR}). As the error decreases, the dominant bottleneck changes, leading to distinct operating regimes and budget contributions. Thresholds and budget contributions are shown up to multiplicative constants.  See Table~\ref{tab:notation} for symbols.
}
\label{fig:bottleneck}
\end{figure*}

\subsection{Contributions}

The main contributions of this paper are summarized as follows:

\begin{itemize}[leftmargin=*]

\item 
We develop a resource-aware, bottleneck-centric framework for decentralized optimization and introduce a multi-stage formulation that partitions the optimization trajectory into successive error scales.

\item 
We introduce two fundamental quantities,  the gradient-Diversity-to-Network-connectivity Ratio (DNR) and the Gradient-to-Communication-noise Ratio (GCR), and show that they govern the sequence of bottlenecks encountered during optimization.

\item We identify the resulting bottleneck-dominated operating regimes and derive the corresponding budget-optimal operating mode and parameter selection across different error scales.

\item 
We derive explicit budget-complexity bounds that characterize the communication-computation resources required to attain a prescribed accuracy and quantify the contribution of successive bottlenecks to the overall budget.

\end{itemize}

The remainder of this paper is organized as follows. \secref{sysmo} introduces the system model and the multi-stage formulation. \secref{singlestage} develops the single-stage optimization framework and identifies the operating regimes associated with different bottlenecks. \secref{complexity} derives budget-complexity bounds. \secref{numres} presents numerical results, followed by concluding remarks in \secref{concl}. Proofs of the main theoretical results are deferred to the Appendix.
\onecol{\vspace{-5mm}}
\subsection{Notation}

Bold lowercase and uppercase letters denote vectors and matrices, respectively. 
\(\mathbf 0\) and \(\mathbf 1\) are
the all-zeros and all-ones vectors, respectively. \(\I_d\) is the \(d\times d\) identity matrix. Vector/matrix transpose is denoted by \((\cdot)^\top\), \(\|\cdot\|\) denotes the Euclidean norm, and \(\otimes\) denotes the Kronecker product. The gradient of a function \(f\) is denoted by \(\nabla f\), and \(\mathbb E[\cdot]\) denotes expectation. 
The main symbols are summarized in Table~\ref{tab:notation}.

\begin{table*}[t]
\centering
\caption{Summary of symbols}
\label{tab:notation}
\footnotesize
\setlength{\tabcolsep}{2pt}
\setlength{\arrayrulewidth}{0.8pt}
\resizebox{\textwidth}{!}{%
\twocol{\begin{tabular}{>{\raggedleft\arraybackslash}p{1cm} p{7.5cm} | >{\raggedleft\arraybackslash}p{1cm} p{7.5cm} }}
\onecol{\begin{tabular}{>{\raggedleft\arraybackslash}p{1cm} p{8.5cm} | >{\raggedleft\arraybackslash}p{1cm} p{8.5cm} }}
\noalign{\hrule height 0.8pt}
\noalign{\vskip 1pt}
\textbf{Symbol} & \textbf{Definition} &\textbf{Symbol} & \textbf{Definition}\\
\noalign{\hrule height 0.8pt}
$\mu$ & strong-convexity parameter (Assumption \ref{strongas})&$\xstar$ & global optimum (Eqs. \ref{xstar1} \& \ref{xstar}) \\
$L$ & smoothness parameter (Assumption \ref{strongas}) & $\xloc$ & local minima (Eqs. \ref{xloci} \& \ref{xloc})\\
$\kappa$ & condition number $L/\mu$&$\x_t$ & iterate at time $t$ (Eq. \ref{xtdyn})
\\
$N$ & number of agents&
$\W$ & mixing matrix (Eq. \ref{xtdyn} \& Assumption \ref{Was}) \\
$\sigc^2$ & communication noise variance bound (Eq. \ref{xtdyn} \& Assumption \ref{noiseas}) &
$E_t$ & root mean squared error at time $t$ (Eq. \ref{rmse})
\\
$\sigg^2$ & gradient noise variance bound (Eq. \ref{xtdyn} \& Assumption \ref{noiseas})
&
$\err$ & target error for the optimization problem (Eq. \ref{budopt})
\\
$\DNR$ & gradient-Diversity-to-Network-connectivity Ratio (Eq. \ref{DNR})&
$\Eloc$ & distance between local \& global optimal (Eq. \ref{Eloc})
\\
$\GCR$ & Gradient-to-Communication noise Ratio (Eq. \ref{GCR})
&
$\Es$ & error bound at the start of stage $s$ (Eq. \ref{barEs})
\\
$\M$ & stage-wise error decay factor (Eqs. \ref{deltasloc} \& \ref{deltasDGD})
&
$\eta_t,\etas$ & learning stepsize at time $t$ or stage $s$ (Eqs. \ref{xtdyn} \& \ref{ss_stage})
\\
$\bd$ & budget cost of Full DGD iteration
&
$\gamma_t,\gammas$ & consensus stepsize at time $t$ or stage $s$ (Eqs. \ref{xtdyn} \& \ref{ss_stage})
\\
$\bg$ & budget cost of Local SGD iteration &
$\deltas$ & number of iterations of stage $s$ (Eq. \ref{stagelength})
\\
\noalign{\vskip 1pt}
\noalign{\hrule height 0.8pt}
\end{tabular}}
\end{table*}

\section{System Model  and Multi-Stage Formulation}\label{sysmo}

We consider $N$ spatially distributed devices cooperatively solving the optimization problem
\begin{align}
\label{xstar1}
x^\star= \arg\min_{x \in \mathbb{R}^d}
F(x),
\end{align}
where
$F(x)\triangleq\frac{1}{N} \sum_{i=1}^N f_i(x)$ is the global objective,
 $f_i:\mathbb R^d\mapsto\mathbb R$ denotes the local loss function of node $i$, known only to node $i$.
 The nodes solve \eqref{xstar1} via decentralized gradient descent (DGD) over a communication graph. Let \(\mathcal N_i\) denote the set of neighbors of node \(i\) (including $i$ itself), and \(w_{ij}\) be the corresponding mixing weights, with \(w_{ij}=0\) whenever \(j\notin \mathcal N_i\). 
 At iteration $t$, node $i$
 updates its local parameter vector $x_{i,t}$
 by forming a weighted combination of the received neighbor models using the corresponding mixing weights, corrupted by communication noise, while simultaneously computing a noisy local gradient.
 This process yields the update
\begin{align}
x_{i,t+1}
= &(1-\gamma_t)x_{i,t}
+ \gamma_t \Big( \sum_{j=1}^N w_{ij} x_{j,t} + \epsilon_{\mc,i,t} \Big)
\twocol{\nn\\}&\onecol{\hspace{-5em}}
- \eta_t\big(\nabla f_i(x_{i,t}) + \epsilon_{\mg,i,t}\big),\ \forall i,
\label{local_update}
\end{align}
where:
\(\epsilon_{\mc,i,t}\) denotes the communication noise affecting the aggregation step;
\(\epsilon_{\mg,i,t}\) denotes the noise
injected during gradient computation, e.g., due to minibatch sampling;
$\gamma_t \in [0,1]$ is the consensus stepsize, controlling the rate of information diffusion across the network; \(\eta_t\geq 0\) is the learning  stepsize, controlling the incorporation of local gradient updates.
Defining the stacked vectors
$$
\x_t{\triangleq}\!\!
\begin{bmatrix}
\!x_{1,t}\!\\
\vdots\\
\!x_{N,t}\!
\end{bmatrix}\!\!,
\nabla f(\x_t){\triangleq}\!\!
\begin{bmatrix}
\!\nabla f_1(x_{1,t})\!\\
\vdots\\
\!\nabla f_N(x_{N,t})\!
\end{bmatrix}\!\!,
\boldsymbol{\epsilon}_{\mc,t}{\triangleq}\!\!
\begin{bmatrix}
\!\epsilon_{\mc,1,t}\!\\
\vdots\\
\!\epsilon_{\mc,N,t}\!
\end{bmatrix}\!\!,
\boldsymbol{\epsilon}_{\mg,t}{\triangleq}\!\!
\begin{bmatrix}
\!\epsilon_{\mg,1,t}\!\\
\vdots\\
\!\epsilon_{\mg,N,t}\!
\end{bmatrix}\!\!,
$$
the recursion \eqref{local_update} can be compactly expressed as
\begin{align}
\label{xtdyn}
\x_{t+1}
=&
(1-\gamma_t)\x_t
+\gamma_t[(\W\otimes\I_d)\x_t+\boldsymbol{\epsilon}_{\mc,t}]
\twocol{\nn\\}&\onecol{\hspace{-6em}}
-\eta_t(\nabla f(\x_t)+\boldsymbol{\epsilon}_{\mg,t}),
\end{align}
where \(\W=[w_{ij}] \in \mathbb{R}^{N\times N}\) is the mixing matrix (see Assumption \ref{Was}).
In other words, DGD combines a consensus step (a weighted average of the models of the neighbors), via inter-agent model exchanges, controlled by a consensus stepsize $\gamma_t$, with a local gradient descent step, controlled by a learning stepsize $\eta_t$. Both communications and gradient computations are affected by noise, characterized in Assumption \ref{noiseas}.

\begin{remark}
We adopt an abstract model of the communication channel, captured by the weight matrix $\W$ and the communication noise $\boldsymbol{\epsilon}_{\mc,t}$. For instance, in quantized DGD \cite{8786146}, $\W$ is typically constructed using Metropolis-Hastings weights \cite{6854643}, with neighbors defined by physical proximity. In this setting, $\boldsymbol{\epsilon}_{\mc,t}$ corresponds to the error induced by dithered quantization.
In DGD with link failures \cite{10947567,9772390}, $\W$ admits a similar interpretation, while $\boldsymbol{\epsilon}_{\mc,t}$ captures the mismatch between the instantaneous link realizations (successful or failed transmissions) and their average success probabilities.
In non-coherent over-the-air DGD \cite{10680589,michelusi26}, mixing occurs directly over the wireless channel by exploiting its average energy superposition property. In this case, $\W$ is dictated by average path loss conditions, whereas $\boldsymbol{\epsilon}_{\mc,t}$ arises from channel fading, non-coherent signal superposition, receiver noise, and interference.
\end{remark}

We assume that both communication and gradient computation consume resources, such as latency, energy, or a combination thereof, and therefore associate a cost with each operation. At each iteration, the network selects an operating mode specifying whether communication, computation, or both are performed. All modes are described by the unified recursion \eqref{xtdyn}, obtained by appropriately activating or suppressing its communication and computation components:
\begin{itemize}[leftmargin=*]
\item \textbf{Local SGD ($\gamma_t = 0$):}  
In this case, the update reduces to
\begin{align}
\x_{t+1}
=
\x_t
-
\eta_t \nabla f(\x_t)
-
\eta_t \boldsymbol{\epsilon}_{\mg,t},
\label{localSGDeq}
\end{align}
i.e., nodes perform a local gradient descent step without communication.  
The cost of this operation is denoted by $\bg$.

\item \textbf{Consensus ($\eta_t = 0$):}  
In this case, the update becomes
\[
\x_{t+1}
=
(1-\gamma_t)\x_t
+
\gamma_t (\mathbf{W}\otimes\I_d)\x_t
+
\gamma_t \boldsymbol{\epsilon}_{\mc,t},
\]
corresponding to one consensus round without a local gradient step,
with cost $b_{\mc}$.

\item \textbf{Full DGD ($\eta_t > 0$ and $\gamma_t > 0$):}  
In this case, the update includes both gradient computation and communication.  
The associated cost is denoted by $\bd$.

\end{itemize}

For example, if the cost represents latency, then
$\bd = \max\{b_{\mc}, \bg\}$ since communication and computation can be executed in parallel with DGD. 
On the other hand, if the cost represents energy consumption, then
$\bd= b_{\mc} + \bg$, due to the overall additive contribution of computation and communication energy costs. More generally, we assume that
$\bd\geq\max\{b_{\mc}, \bg\}$.

The network is given a target accuracy level \(\err\). During the optimization process, the network scheduler\footnote{The scheduler is introduced as a coordination abstraction. In practice, its decisions could be communicated over a low-rate, long-range control channel (e.g., an LPWAN \cite{7721743}). Since the proposed framework updates the operating mode only at stage transitions, such coordination is required only infrequently.} selects at each iteration one of the operating modes described above, together with the corresponding communication and learning stepsizes. The objective is to attain the prescribed accuracy using the minimum cumulative budget, leading to the optimization problem
\begin{align}
\label{budopt}
\min_{\{\gamma_t \ge 0,\, \eta_t \ge 0\}_{t=0}^{T-1},\, T \ge 0}
\quad \sum_{t=0}^{T-1} b_t
\quad\text{s.t.}\quad
E_T \le \err,
\end{align}
where \(b_t\) is the cost incurred at iteration \(t\), determined by the selected operating mode, \(T\) is the total number of iterations, and
\begin{align}
\label{rmse}
E_t \triangleq \sqrt{\mathbb E\!\left[\|\x_t-\xstar\|^2\right]}
\end{align}
denotes the root mean squared error (RMSE) between the iterates at time \(t\) and the global optimum \(\xstar \triangleq \mathbf 1_N \otimes x^\star\).
 In other words, we seek the most budget-efficient path to a target accuracy level \(\err\). 

Directly solving \eqref{budopt} is generally intractable because the exact RMSE $E_t$
 cannot be characterized in closed form as a function of the decision variables $\eta_t,\gamma_t$
 and the horizon $T$. Rather, it depends on the optimization dynamics determined by the objective functions, network topology, noise statistics, and stepsize schedule. We therefore analyze a computable upper bound  $E_t \le \bar E_t$.
Unlike the exact RMSE,  $\bar E_t$ admits an explicit characterization in terms of the problem parameters, including the strong convexity and smoothness constants (Assumption~\ref{strongas}), network connectivity (Assumption~\ref{Was}), initialization error, objective heterogeneity, noise statistics (Assumption~\ref{noiseas}), and the stepsize schedule. Moreover, 
$\bar E_t$
 provides a rigorous performance guarantee, ensuring that the actual RMSE $E_t$ remains below the prescribed bound.

Nevertheless, even after replacing \(E_t\) by \(\bar E_t\), the resulting budget-allocation problem remains analytically intractable. The bound \(\bar E_t\) depends on the entire optimization trajectory, so decisions made at a given iteration influence all future rounds, resulting in a highly coupled dynamic optimization problem.
A key observation is that the error typically evolves gradually, implying that the quantities governing convergence vary only slowly over time. Consequently, the budget-optimal communication-computation tradeoff is expected to remain nearly constant over a range of iterations.
 This observation suggests a multi-stage formulation. Rather than adapting the algorithm at every iteration, we partition the optimization trajectory into stages associated with different error scales. Within each stage, the error is allowed to vary only by a prescribed factor, while the communication and learning stepsizes remain fixed. The algorithm parameters are updated only when the error has decreased sufficiently to alter the dominant factors governing the convergence behavior. This approximation makes the design problem tractable while preserving the dominant mechanisms that determine the budget-optimal strategy.
 
 Specifically, we partition the optimization horizon into \(S\) stages. The \(s\)th stage starts at iteration \(k_s\), with \(k_0=0\) for the initial stage $s=0$,
and has duration $\delta_s$, so that the next stage starts at iteration
\begin{align}
\label{stagelength}
k_{s+1}= k_s+\delta_s.
\end{align}
Let \(\M>1\) denote the target error-reduction factor per stage. We define the error bound at the beginning of stage \(s\) as
\begin{align}
\label{barEs}
E_{k_s}
\le
\frac{\E{0}}{\M^s}
\triangleq
\Es.
\end{align}
Accordingly, the objective of stage $s$ is to reduce the error bound to
\begin{align}
\label{phidef}
E_{k_{s+1}}
\le
\frac{\Es}{\M}
\triangleq
\Esn.
\end{align}
Within stage \(s\), the stepsizes are held constant,
\begin{align}
\label{ss_stage}
\eta_t=\etas,\qquad
\gamma_t=\gammas,\qquad
t=k_s,\ldots,k_{s+1}-1,
\end{align}
and the operating mode, stepsizes \((\etas,\gammas)\), and stage length $\delta_s$,
are chosen to attain \eqref{phidef} with minimum budget cost.

Since the error bound decreases only from $\Es$ to $\Es/\M$ during a stage, the optimization process remains within the same error scale. Consequently, the relative importance of initialization, objective heterogeneity, limited network connectivity, and noise is expected to remain nearly unchanged, making parameter adaptation within the stage of limited benefit. The operating mode and stepsizes are therefore updated only at stage boundaries, where the error has decreased sufficiently for the optimal communication-computation tradeoff to change.

In the remainder of the paper, we restrict attention to \emph{Local SGD} (computation only) and \emph{Full DGD} (simultaneous communication and computation), leaving the treatment of consensus-only rounds to future work. Since consensus-only updates primarily reduce disagreement rather than optimality error, their analysis requires explicitly tracking the coupled evolution of consensus and optimality errors. While of independent interest, this would substantially complicate the analysis and is therefore beyond the scope of the present work.

To develop our multi-stage budget-limited framework,
we begin by introducing several assumptions and definitions
used in the analysis. 
\onecol{\vspace{-3mm}}
\subsection{Assumptions}

We make the following assumption on each local function $f_i$ and
 the mixing matrix $\W$.
 \begin{assumption}\label{strongas}
Each local objective function $f_i:\mathbb{R}^d\to\mathbb{R}$ is
$\mu$-strongly convex and has Lipschitz continuous gradients with parameter $L$ (i.e., it is $L$-smooth), where $0<\mu\le L$.
 It follows that $F(x)$ is also $\mu$-strongly convex and $L$-smooth, and the global minimizer $x^\star$ is unique.
\end{assumption}
  
\begin{assumption}
\label{Was}
The mixing matrix $\W\in\mathbb{R}^{N\times N}$ is
entry-wise  non-negative ($\W\ge 0$), symmetric ($\W^\top=\W$), and doubly stochastic
($\W\;\1=\1$ and $\1^\top{=}\1^\top\;\W$). Furthermore, the communication graph induced by the positive entries of \(\W\) is connected.
By the Perron-Frobenius Theorem \cite[Th.~8.4.4]{matrixanalysis}, the eigenvalues of $\W$ are real and satisfy
\[
1=\lambda_1>\lambda_2\ge\lambda_3\ge\dots\ge\lambda_N\ge -1.
\]

\end{assumption}

Let \(\{\mathcal F_t\}_{t\ge 0}\) denote the natural filtration generated by the iterates and all randomness up to time \(t\), i.e.,
\[
\mathcal F_t = \sigma\!\left(\x_0, 
\boldsymbol{\epsilon}_{\mc,0}, \boldsymbol{\epsilon}_{\mg,0},
\dots,
\boldsymbol{\epsilon}_{\mc,t-1}, \boldsymbol{\epsilon}_{\mg,t-1}
\right).
\]
Note that $\x_t$ is $\mathcal F_t$-measurable, i.e.,
it is fully determined by the history up to time  $t$, through the dynamics in \eqref{xtdyn}.
We impose the following assumptions on the noise processes.

\begin{assumption}
\label{noiseas}
The communication and gradient noises satisfy:
\begin{align}
&\mathbb E\left[\boldsymbol{\epsilon}_{\mc,t}\mid \mathcal F_t\right] = \mathbf 0,\qquad
\mathbb E\left[\|\boldsymbol{\epsilon}_{\mc,t}\|_2^2 \mid \mathcal F_t\right]
\le \sigc^2;\\
&\mathbb E\left[\boldsymbol{\epsilon}_{\mg,t}\mid \mathcal F_t\right] = \mathbf 0, \qquad
\mathbb E\left[\|\boldsymbol{\epsilon}_{\mg,t}\|_2^2 \mid \mathcal F_t\right]
\le \sigg^2 .
\end{align}
\end{assumption}
Here, $\sigc^2$ and $\sigg^2$ characterize the variance of the communication and computation noises, respectively. 
We do not require the communication and gradient noises to be conditionally uncorrelated,
i.e., $\mathbb E\left[\boldsymbol{\epsilon}_{\mg,t}^\top\cdot
\boldsymbol{\epsilon}_{\mc,t}
\mid \mathcal F_t\right] $ may be non-zero.
Unbiasedness of the gradient noise is guaranteed by a suitable minibatch sampling strategy.
On the other hand, unbiasedness of the communication noise can be achieved by suitable design of the encoding scheme, e.g.  via dithered quantization \cite{8786146},
and of the communication protocol (see for instance \cite{10680589,michelusi26}).
Assumption~\ref{noiseas} adopts a uniformly bounded noise variance model. Extensions to signal-dependent noise variance models (e.g., \cite{10947567}) are left for future work.


\subsection{Definitions}

We define the condition number of the optimization problem as
\(
\kappa \triangleq \frac{L}{\mu}
\)
($\ge 1$
since $L \ge \mu$). It measures the intrinsic difficulty of the optimization problem, with larger values corresponding to more ill-conditioned objectives, hence slower convergence of gradient-based methods.

We define the local minimizer at each agent as
\begin{align}
\label{xloci}
x_i^{\mathrm{loc}}=\arg\min_{x\in\mathbb R^d} f_i(x).
\end{align}
Under Assumption~\ref{strongas}, each $x_i^{\mathrm{loc}}$ exists and is unique.
We stack the local minimizers across the network into the vector
\begin{align}
\label{xloc}
\xloc =
\left[
x_1^{\mathrm{loc}\top},
x_2^{\mathrm{loc}\top},
\dots,
x_N^{\mathrm{loc}\top}
\right]^\top.
\end{align}
We define the global optimum stacked over the network as
\begin{align}
\label{xstar}
\xstar\triangleq \1_N\otimes x^\star.
\end{align}
We define the distance between local minima and the global optimum as
\begin{align}
\label{Eloc}
\Eloc \triangleq \|\xloc - \xstar\|.
\end{align}
We define the \emph{gradient heterogeneity} at the global optimum as
\[
\|\nabla f(\xstar)\|
=
\Big(\sum_{i=1}^N \|\nabla f_i(x^\star)\|^2 \Big)^{1/2}.
\]
Notably, if all local minima coincide with the global optimum (i.e., $\xloc=\xstar$), then
$\Eloc = 0$
and $\nabla f_i(x^\star)=\mathbf 0,\ \forall i,$
which implies $\|\nabla f(\xstar)\|=0$.
Hence, both \(\Eloc\) and \(\|\nabla f(\xstar)\|\) quantify the degree of \emph{objective heterogeneity} across the network. They are closely related as\footnote{Seen by using \(\nabla f(\xloc)=\mathbf 0\),
\(
\|\nabla f(\xstar)\|
=
\|\nabla f(\xstar)-\nabla f(\xloc)\|
\), followed by
 strong convexity and smoothness.}
\begin{align}
\mu\;\Eloc
\le
\|\nabla f(\xstar)\|
\le
L\;\Eloc.
\label{objheter}
\end{align}

We define the \emph{gradient-Diversity-to-Network-connectivity Ratio} (DNR) as
\begin{align}
\DNR\triangleq\Big(\frac{2\kappa}{L+\mu}\frac{\|\nabla f(\xstar)\|}{1-\lambda_2}\Big)^2,
\label{DNR}
\end{align}
where $\lambda_2$ is the second-largest eigenvalue of the mixing matrix $\W$ (Assumption \ref{Was}).
The DNR captures the relative difficulty of mitigating objective heterogeneity (captured by $\|\nabla f(\xstar)\|$) through DGD updates across the network (captured by $1-\lambda_2$).
Large values of $\DNR$ arise when objective heterogeneity is high, or the network is poorly connected ($\lambda_2$ close to $1$), making disagreement mitigation across agents more challenging.
Furthermore, a higher condition number $\kappa$ makes the problem harder to solve over a network, yielding larger DNR.

Finally, we define the \emph{Gradient-to-Communication noise Ratio} (GCR)
\begin{align}
\label{GCR}
\GCR
\triangleq
\frac{4\sigg^2}{(L+\mu)^2\sigc^2}.
\end{align}
It compares the relative impact of stochastic gradient noise and communication noise on the evolution of the iterates. 
Large values of \(\GCR\) indicate that gradient noise is relatively stronger than communication noise, whereas small values indicate a comparatively larger influence of the latter.
The scaling factor $4/(L+\mu)^2$ accounts for the effect of the learning and consensus stepsizes on the propagated noise. In particular, the ratio $\mathrm{var}(\eta\,\boldsymbol{\epsilon}_{\mg,t})/\mathrm{var}(\gamma\,\boldsymbol{\epsilon}_{\mc,t})$ coincides with the definition above when the stepsizes take their largest admissible values $\gamma{=}1$ and $\eta{=}2/(L+\mu)$, required for convergence of gradient descent with strongly convex objectives \cite{10947567}.

\section{Single-Stage Bottleneck Analysis}\label{singlestage}

With these definitions, we now analyze the convergence properties of a generic stage, in which the learning and consensus stepsizes are held fixed. These results are then combined in the multi-stage formulation developed in \secref{complexity}.

To initialize the stage-wise recursion, we require a computable upper bound on the initial error, 
$E_0\leq\E{0}$, whose derivation is deferred to the Appendix.
We therefore consider a generic stage $s\geq 0$, beginning at iteration $k_s$.
 The stage starts with error bound $E_{k_s}\leq\Es$
  and aims to reduce it to $E_{k_{s+1}}\leq\Es/\M$, with minimum budget cost. We analyze the Local-SGD and Full-DGD operating modes separately.

\subsection{Local SGD}
\label{locsgdan}
When communication is disabled (\(\gamma_t=0\)), 
the one-step update with fixed $\eta$ reduces to
\[
\mathbf x_{t+1}=\mathbf x_t-\eta\nabla f(\mathbf x_t)-\eta\boldsymbol{\epsilon}_{\mg,t}.
\]
Since no information is exchanged among agents, each node optimizes its local objective independently. Consequently, if this mode is maintained, the algorithm converges toward the local minima \(\xloc\), yielding a residual error $\xloc-\xstar$ due to objective heterogeneity.
The following lemma bounds the RMSE after $\delta$ iterations of a Local-SGD stage with fixed learning stepsize.

\begin{lemma}\label{lem:localSGD}
Consider a Local-SGD stage with learning stepsize \(0\le\eta\le2/(L+\mu)\). After \(\delta\) iterations,
the RMSE satisfies
\[
E_{k_s+\delta}
\le 
3\max\Big\{
(1-\eta\mu)^{\delta}\Es,
\frac{\sqrt{\eta(\kappa+1)}}{\sqrt{2L}}\sigg,\twocol{\]\[}
(1+(1-\eta\mu)^{\delta})\Eloc
\Big\}.
\]
\end{lemma}
\begin{proof}
See Appendix \ref{proofoflem:localSGD}.
\end{proof}
This bound reveals that the error after $\delta$ rounds of 
Local SGD is the worst among three terms:
\begin{enumerate}[leftmargin=*]
  \item  the term \((1-\eta\mu)^{\delta}\Es\), accounting for the
   impact of the initialization error of stage $s$;
  \item the propagated gradient-noise term, \(\propto\sigg\),
  \item the mismatch between local and global optima, \(\Eloc\), due to objective heterogeneity.
\end{enumerate}
The worst among these error terms acts as a 
 \emph{bottleneck} in the error dynamics.

\medskip\noindent\textbf{Reducing the error by a factor $\M$:}
We wish to choose the stepsize \(\eta\) and the length \(\delta\) of stage $s$, 
so that  \(E_{k_s+\delta}\le \Es/\M\) (the error is reduced by a factor $\M$ by the end of the stage) with minimum budget cost.
Using the bound found above, it is sufficient to impose
\begin{subequations}\label{eq:three_conditions}
\begin{empheq}[left=\empheqlbrace]{align}
\label{eq:three_conditions.a}&(1-\eta\mu)^{\delta}\Es \le \frac{\Es}{3\M},\ \text{(initialization constraint)};\\
\label{eq:three_conditions.b}&\frac{\sqrt{\eta(\kappa+1)}}{\sqrt{2 L}}\sigg \le\frac{\Es}{3\M},\ \text{(gradient noise constraint)};\\
\label{eq:three_conditions.c}&\Big(1{+}\frac{1}{3\M}\Big)\Eloc \le \frac{\Es}{3\M},\  \text{(heterogeneity constraint)};
\end{empheq}
\end{subequations}
where we used $(1-\eta\mu)^{\delta}\leq 1/(3\M)$ from \eqref{eq:three_conditions.a} to simplify \eqref{eq:three_conditions.c}.
These constraints can be interpreted as follows:
\begin{itemize}[leftmargin=*]
  \item The \emph{initialization constraint} term determines how many local steps \(\delta\) are required to diminish the effect of stage initialization below the required threshold (given \(\eta\));
  \item The \emph{gradient noise constraint}, controlled by $\eta$, ensures that the accumulated gradient noise does not exceed the desired target;
  \item The \emph{heterogeneity constraint} captures the bias due to objective heterogeneity (mismatch between the local and global minimizers). It is feasible only if $\Es\geq (3\M+1)\Eloc$.
\end{itemize}

The constraint \eqref{eq:three_conditions.a} requires
$\delta \ge \frac{\ln(3\M)}{-\ln(1-\eta\mu)}.$
From \eqref{eq:three_conditions.b}, coupled with the stepsize condition of Lemma \ref{lem:localSGD}, we obtain
\begin{align}
\label{etaconstr}
\eta \le 
\frac{2}{L+\mu}
\min\left\{1\ ,\ \mu L\Big(\frac{\Es}{3\M\sigg}\Big)^2\right\}.
\end{align}
Finally, if \eqref{eq:three_conditions.c} is violated ($\Es<(3\M+1)\Eloc$), then Local SGD is
unfeasible. In this case, we need communication (studied in the next section) to
mitigate the effect of objective heterogeneity and further drive the error down.

The next stage starts at time  $k_{s+1}=k_s+\delta$, with the corresponding error bounded by
 $\Esn=\Es/\M$.
Under this stepsize design, we can also bound
the error for the iterates $t=k_s\dots,k_{s+1}$, by combining the Local-SGD bound
 \eqref{errb} with $(1-\eta\mu)^{t-k_s}\leq 1$ as
\begin{align}
&E_{t}
\le \Es+\frac{\sqrt{\eta(\kappa+1)}}{\sqrt{2 L}}\sigg+2\Eloc
\twocol{\nn\\}&\onecol{\hspace{-35mm}}
\quad\le 
\frac{(3\M)^2+12\M+1}{3\M(3\M+1)}\Es,
\label{EellboundSGD}
\end{align}
where we used 
\eqref{eq:three_conditions.b}-\eqref{eq:three_conditions.c} in the last inequality.

\noindent\textbf{Minimum stage-wise budget cost:}
Under the condition $\Es \ge(3\M+1)\Eloc$, reducing the RMSE by a factor $\Phi$ using Local SGD is feasible.
The budget cost of performing $\delta$ local gradient steps is $\delta \cdot \bg$.
Hence, the budget in the $s$th stage is minimized by choosing $\delta$ as small as possible, yielding
\begin{align}
\label{deltasloc}
\deltas \triangleq \left\lceil
\frac{\ln(3\M)}{-\ln(1-\eta\mu)}
\right\rceil.
\end{align}
Since this expression is decreasing in $\eta$, minimizing the budget cost amounts to
maximizing $\eta$. Coupled with \eqref{etaconstr},
this yields the optimal stepsize in stage $s$,
\begin{align}
\label{optetalocsgd}
\etas
\triangleq
\frac{2}{L+\mu}
\min\left\{1\ ,\ \mu L\Big(\frac{\Es}{3\M\sigg}\Big)^2\right\}.
\end{align}

Plugging this expression into \eqref{deltasloc}, we obtain the minimum budget cost 
required to reduce the error by a factor $\M$ as
\begin{align}
\label{budgbo}
\bg\cdot\deltas=\bg\cdot \left\lceil
\frac{\ln(3\M)}{-\ln(1-\etas\mu)}
\right\rceil\leq \bg\cdot \left\lceil
\frac{\ln(3\M)}{\etas\mu}
\right\rceil,
\end{align}
where we used $-\ln(1-\etas\mu)\geq \etas\mu$ in the last inequality.

This analysis further reveals two regimes of interest.

\subsubsection{Initialization-dominated regime of Local SGD}
\label{SGDr1}
If
\begin{align}
\label{E0init}
\Es
\ge
\max\!\left\{
\frac{3\M\sigg}{\sqrt{\mu L}}\ ,\;
(3\M+1)\Eloc
\right\},
\end{align}
then both the gradient-noise and heterogeneity constraints
\eqref{eq:three_conditions.b}-\eqref{eq:three_conditions.c}
 are inactive, and the optimal stepsize specializes to
\(
\etas=\frac{2}{L+\mu}.
\)
Replacing this expression into \eqref{budgbo}, we bound
the stage-wise budget cost
as
\begin{align}
\bg \cdot \deltas
\leq
\bg \cdot 
\left\lceil
\ln(\sqrt{3\M})(\kappa+1)
\right\rceil.
\label{locsgdinit}
\end{align}
This corresponds to the \emph{initialization-dominated regime}, where the initial stage error $ \Es$ is sufficiently large. In this regime, the contraction term
\(
(1-\eta\mu)^{\delta}\Es
\)
acts as the dominant factor governing error reduction, while gradient-noise and objective heterogeneity remain negligible relative to $\Es$.
Consequently, the behavior resembles that of \emph{noiseless centralized gradient descent}: the error decreases geometrically, and the cost required to reduce the error by a constant factor $\M$ is independent of the gradient noise and heterogeneity metric.
Instead, this cost grows proportionally to the condition number
$\kappa$, capturing the difficulty of the optimization problem.


\subsubsection{Gradient-noise-dominated regime  of Local SGD}
\label{SGDr2}
If the error bound satisfies
\begin{align}
\frac{3\M\sigg}{\sqrt{\mu L}}
>
\Es
\geq
(3\M+1)\Eloc,
\label{xcv}
\end{align}
then the gradient-noise constraint \eqref{eq:three_conditions.b} becomes active, and the optimal stepsize is
\[
\etas=\frac{2\mu L}{L+\mu}\Big(\frac{\Es}{3\M\sigg}\Big)^2.
\]

This regime is feasible as long as the interval in \eqref{xcv} is non-empty, i.e.,
\[
\sigg
>
\Big(1+\frac{1}{3\M}\Big)
\sqrt{\mu L}\Eloc,
\]
corresponding to a setting where gradient noise is significant.

Plugging the expression of $\etas$ into \eqref{budgbo}, we bound the stage-wise budget cost 
as
\begin{align}
\label{deltagradsgd}
\bg \cdot \deltas
\leq
\bg \cdot 
\Big\lceil
9\M^2\ln(3\M)\Big(\frac{\sigg}{\mu\Es}\Big)^2
\Big\rceil.
\end{align}
This corresponds to the \emph{gradient-noise-dominated regime}, where the noise term
\(
\propto \sigg
\)
acts as the primary bottleneck in the error bound, while the heterogeneity term remains small relative to $\Es$. 
Hence, the budget cost
 grows proportionally to the gradient noise variance $\sigg^2$.

\subsection{Full DGD}
\label{dgdan}
With both local computations and communications enabled and fixed stepsizes $\gamma,\eta> 0$,
the update becomes
\begin{align}
\x_{t+1}
=&
(1-\gamma)\x_t
+
\gamma (\mathbf{W}\otimes\I_d)\x_t\twocol{\nn\\}&\onecol{\hspace{-30mm}}
-
\eta \nabla f(\x_t)
+
\gamma \boldsymbol{\epsilon}_{\mc,t}
-
\eta \boldsymbol{\epsilon}_{\mg,t}.
\label{contrct}
\end{align}
The following lemma bounds the RMSE after $\delta$ iterations of a Full-DGD stage.

\begin{lemma}\label{lem:fullDGD}
Consider a Full-DGD stage with learning stepsize
$0<\eta\leq\frac{2-\gamma(1-\lambda_N)}{L+\mu}$ and
consensus stepsize $\gamma\in(0,1]$. After \(\delta\) iterations,
the RMSE satisfies
\begin{align}
&E_{k_s+\delta}\leq 
4\max\Big\{(1-\eta\mu)^{\delta}\Es\ ,\ 
\frac{\sqrt{\eta(\kappa+1)}}{\sqrt{2 L}}\sigg\ ,\ \nn\\&
\frac{\gamma}{\sqrt{\eta}}\frac{\sqrt{\kappa+1}}{\sqrt{2 L}}\sigc\ ,\ 
(1+(1-\eta\mu)^{\delta})\frac{\eta}{\gamma}\frac{L+\mu}{2}\sqrt{\DNR}
\Big\}.
\label{dgdbound}
\end{align}
\end{lemma}
\begin{proof}
See Appendix \ref{proofoflem:fullDGD}.
\end{proof}

Similar to the analysis of Local SGD, the above bound shows that the error after $\delta$ rounds of Full DGD is determined by the largest among four terms:
\begin{enumerate}[leftmargin=*]
\item the term $(1-\eta\mu)^{\delta}\Es$, capturing the effect of the initial error at the beginning of stage $s$;

\item the propagated gradient-noise term,  $\propto\sigg$;

\item the propagated communication-noise term, $\propto\sigc$;

\item the bias due to the mismatch between the fixed point $\hat{\x}$ which Full DGD aims to converge to and the desired global optimum $\xstar$, scaling with the DNR.
\end{enumerate}
The largest of these contributions acts as a \emph{bottleneck} in the error dynamics.

\medskip\noindent\textbf{Reducing the error by a factor $\M$:}
We seek to choose the learning stepsize \(\eta\), the consensus stepsize \(\gamma\), and the 
length $\delta$ of stage \(s\), such that
\(
E_{k_s+\delta}\le \frac{\Es}{\M}
\)
 with minimum budget cost.
Using \eqref{dgdbound}, it suffices to impose
\begin{subequations}\label{eq:four_conditions}
\begin{empheq}[left=\empheqlbrace]{align}
\label{eq:four_conditions.a}&(1-\eta\mu)^{\delta}\Es \le \frac{\Es}{4\M},\  \text{(initialization constraint)};\\
\label{eq:four_conditions.b}&\frac{\sqrt{\eta(\kappa+1)}}{\sqrt{2 L}}\sigg \le \frac{\Es}{4\M},\  \text{(gradient noise constrain)};\\
\label{eq:four_conditions.c}&\frac{\gamma}{\sqrt{\eta}}\frac{\sqrt{\kappa+1}}{\sqrt{2 L}}\sigc \le \frac{\Es}{4\M},\ \text{(comm. noise constraint)};\\
\label{eq:four_conditions.d}&\Big(1{+}\frac{1}{4\M}\Big)\frac{\eta}{\gamma}\frac{L{+}\mu}{2}\sqrt{\DNR} \le \frac{\Es}{4\M},\ \text{(FP constraint)}.
\end{empheq}
\end{subequations}
Above, we used the bound $(1-\eta\mu)^{\delta}\le 1/(4\M)$ implied by
\eqref{eq:four_conditions.a} to simplify \eqref{eq:four_conditions.d}.
These constraints admit the following interpretation:
\begin{itemize}[leftmargin=*]
\item The \emph{initialization constraint} determines the number of iterations \(\delta\) required to sufficiently attenuate the effect of the initial stage error (for a given stepsize \(\eta\)).

\item The \emph{gradient-noise}
and \emph{communication-noise constraints}
 limit the learning stepsize \(\eta\) 
 and the stepsize ratio \(\gamma/\sqrt{\eta}\), respectively,
 so that the propagated noise remains below the target error.

\item The \emph{fixed point (FP) constraint} captures the intrinsic bias of DGD due to objective heterogeneity and limited connectivity, 
 controlled by the stepsize ratio \(\eta/\gamma\).
\end{itemize}

The next stage starts at time  $k_{s+1}=k_s+\delta$, with the corresponding error bounded by
 $\Esn=\Es/\M$.
Under this stepsize design, we can also bound
the error for the iterates $t=k_s\dots,k_{s+1}$
 using the Full-DGD bound \eqref{dgdboun2} with $(1-\eta\mu)^{t-k_s}\leq 1$ as
\begin{align}
&E_{t}
\le \Es{+}\frac{\sqrt{\eta(L+\mu)}}{\sqrt{2\mu L}}\sigg
{+}\frac{\gamma\sqrt{L+\mu}}{\sqrt{2\eta\mu L}}\sigc
{+}\frac{\eta}{\gamma}(L+\mu)\sqrt{\DNR}
\twocol{\nn\\}&
\le 
\frac{8\M^2+10\M+1}{2\M(4\M+1)}\Es,
\label{dgdalpha}
\end{align}
where in the last inequality we used \eqref{eq:four_conditions.b}-\eqref{eq:four_conditions.d}.

\noindent\textbf{Minimum stage-wise budget cost:}
The budget cost of performing $\delta$ Full-DGD steps is $\delta \cdot \bd$.
Hence,
minimizing $\delta \cdot \bd$ is equivalent to minimizing $\delta$, yielding
from \eqref{eq:four_conditions.a}
\begin{align}
\label{deltasDGD}
\deltas \triangleq \left\lceil
\frac{\ln(4\M)}{-\ln(1-\eta\mu)}
\right\rceil.
\end{align}
Since this expression is decreasing in $\eta$, minimizing the budget cost amounts to 
maximizing $\eta$ with respect to $(\eta,\gamma)$, subject to the constraints \eqref{eq:four_conditions.b}-\eqref{eq:four_conditions.d} and the stepsize conditions
\begin{align}
\label{origcond}
0<\gamma\leq 1,\ 0<\eta\leq \frac{2-\gamma(1-\lambda_N)}{L+\mu},
\end{align}
required in Lemma \ref{lem:fullDGD}.
In the following, we impose slightly stronger conditions
\begin{align}
\label{newcond}
0 < \gamma \le \frac{1}{2}, 
\qquad
0 < \eta \le \frac{1}{L+\mu}.
\end{align}
While the original constraints could be retained, they considerably complicate the analysis. The proposed conditions imply the original ones and are at most a factor 2 more conservative. Since smaller values of $\eta$ and $\gamma$ are often desirable in the presence of gradient and communication noise, this simplification incurs only a minor loss of generality while yielding a much cleaner analysis.
 The following lemma characterizes the budget-optimal learning and consensus stepsizes.
\begin{lemma}
\label{lem:fullDGDoptsteps}
The budget-optimal learning and consensus stepsizes that guarantee
\(E_{k_{s+1}}\le\Es/\M\) with minimum budget cost, subject to constraints \eqref{newcond},
are
\begin{align}
\label{optgam}
\gammas\triangleq&
\min\Bigg\{
\frac{\kappa}{(\kappa+1)^2}\frac{\Es^3}{4\M^2(4\M+1)\sigc^2\sqrt{\DNR}}
,\twocol{\nn\\}&\onecol{\hspace{-10mm}}
\frac{L}{\kappa+1}\frac{\Es^2}{8\M^2\sigg\sigc},
\frac{\sqrt{2\kappa}}{\kappa+1}\frac{\Es}{4\M\sigc},
\frac{1}{2}
\Bigg\}.
\end{align}
\twocol{
\begin{align}
\etas=
&\frac{1}{L+\mu}\min\Bigg\{
\underbrace{\Big(\frac{\sqrt{\kappa}}{\kappa+1}\frac{\Es^2}{\sqrt{2}\M(4\M+1)\sigc\sqrt{\DNR}}\Big)^2}_{\text{(a)}},\nn\\&
\underbrace{2\mu L\Big(\frac{\Es}{4\M\sigg}\Big)^2}_{\text{(c)}},
\underbrace{\frac{\Es}{(4\M+1)\sqrt{\DNR}}
}_{\text{(e)}},
\underbrace{\vphantom{\frac{\Es}{(4\M+1)\sqrt{\DNR}}}1}_{\text{(f)}}
\Bigg\}.
\label{etaopt}
\end{align}}
\onecol{
\begin{align}
\etas{=}
\frac{1}{L{+}\mu}\min\Bigg\{
\underbrace{\Big(\frac{\sqrt{\kappa}}{\kappa{+}1}\frac{\Es^2}{\sqrt{2}\M(4\M{+}1)\sigc\sqrt{\DNR}}\Big)^2}_{\text{(a)}},
\underbrace{2\mu L\Big(\frac{\Es}{4\M\sigg}\Big)^2}_{\text{(c)}},
\underbrace{\frac{\Es}{(4\M{+}1)\sqrt{\DNR}}
}_{\text{(e)}},
\underbrace{\vphantom{\frac{\Es}{(4\M{+}1)\sqrt{\DNR}}}1}_{\text{(f)}}
\Bigg\}.\!\!\!
\label{etaopt}
\end{align}}
\end{lemma}
\begin{proof}
See Appendix \ref{proofoflem:fullDGDoptsteps}.
\end{proof}

Plugging this expression into \eqref{deltasDGD}, we bound the minimum stage-wise budget cost as
\begin{align}
\label{budgbo2}
\bd\cdot\deltas{=}\bd \left\lceil
\frac{\ln(4\M)}{-\ln(1-\etas\mu)}
\right\rceil{\leq} \bd \left\lceil
\frac{\ln(4\M)}{\etas\mu}
\right\rceil,
\end{align}
where we used $-\ln(1-\etas\mu)\geq \etas\mu$ in the last inequality.

The learning stepsize expression in Lemma \ref{lem:fullDGDoptsteps} reveals four operating regimes, depending on which term in the minimum is active, namely (a), (c), (e), or (f).  Each regime is associated with a different bottleneck governing the convergence behavior, studied next.

\subsubsection{Initialization-dominated regime of Full DGD}
\label{DGDr1}
When $\Es$ is sufficiently large, term (f) becomes the smallest among the arguments of the minimum defining $\etas$, yielding
\begin{align}
\label{etainitdgd}
\etas = \frac{1}{L+\mu}.
\end{align}
 The activation of this regime requires that the remaining terms satisfy
$\text{(a)},\text{(c)},\text{(e)} \ge \text{(f)}$ in \eqref{etaopt}. Solving these inequalities yields the following conditions on $\Es$:
\begin{subequations}\label{Einit}
\begin{empheq}[left=\empheqlbrace]{align}
\label{Einit.a}&\Es \ge \dfrac{\sqrt{\kappa+1}}{\sqrt[4]{\kappa}}\sqrt{\M(4\M+1)\sigc}\sqrt[4]{2\;\DNR},\\
\label{Einit.b}&\Es \ge \dfrac{2\sqrt{2}\M\sigg}{\sqrt{\mu L}},\\
\label{Einit.c}&\Es \ge (4\M+1)\sqrt{\DNR}.
\end{empheq}
\end{subequations}
Replacing \eqref{etainitdgd} into \eqref{budgbo2}, 
the stage-wise budget cost becomes
$$\bd \cdot \deltas
\leq
\bd \cdot 
\left\lceil
\ln(4\M)(\kappa+1)
\right\rceil.
$$
This corresponds to the \emph{initialization-dominated regime}, in which the stage initialization error acts as the main bottleneck governing the error evolution. In this regime, the effects of communication noise, gradient noise, objective heterogeneity and network connectivity remain negligible relative to $\Es$, and the convergence behavior is primarily dictated by the contraction induced by the gradient descent dynamics, captured by the condition number $\kappa$. 
In fact, the budget cost is independent of the noise parameters ($\sigg$ and $\sigc$) as well as the DNR.

\subsubsection{DNR-dominated regime of Full DGD}
\label{DGDr2}

When $\Es$ is moderately large, term $\text{(e)}$ becomes the smallest in \eqref{etaopt}, yielding the optimal learning stepsize as
\begin{align}
\label{etaDNRdgd}
\etas
=
\frac{1}{L+\mu}\frac{\Es}{(4\M+1)\sqrt{\DNR}}.
\end{align}
By enforcing $\text{(a)},\text{(c)},\text{(f)}\ge \text{(e)}$ in \eqref{etaopt},
operating in this regime requires
\begin{subequations}\label{cond2}
\begin{empheq}[left=\empheqlbrace]{align}
\label{cond2.a}&\Es\ge\sqrt[3]{\frac{(\kappa+1)^2}{\kappa}\,2\M^2(4\M+1)\sigc^2\sqrt{\DNR}} ,\\
\label{cond2.b}&\Es\ge\frac{8\M^2}{4\M+1}\frac{\sigg^2}{\mu L}\frac{1}{\sqrt{\DNR}},\\
\label{cond2.c}&\Es\le(4\M+1)\sqrt{\DNR}.
\end{empheq}
\end{subequations}
This regime is feasible provided that the interval defined by \eqref{cond2.a}-\eqref{cond2.c} is non-empty, 
yielding
\begin{align}
\label{condx}
\phi\frac{\mu L\,\DNR}{\sigg^2}
\geq
\frac{1}{\min\{\GCR,1\}},
\end{align}
where we defined the scaling factor 
\begin{align}
\label{PHI}
\phi\triangleq\frac{(4\M+1)^2}{8\M^2}\in [2,3.125).
\end{align}

Replacing $\etas$ into \eqref{budgbo2},
the stage-wise budget cost becomes
\begin{align}
\label{deltaDNR}
\bd \cdot \deltas
\leq
\bd\cdot
\left\lceil
\ln(4\M)(4\M+1)(\kappa+1)\frac{\sqrt{\DNR}}{\Es}
\right\rceil.
\end{align}
In this regime, the effects of both gradient and communication noises are small relative to \(\Es\). The convergence behavior is therefore governed by objective heterogeneity and network connectivity, whose combined effect is captured by the DNR. We thus refer to this regime as the \emph{DNR-dominated regime of Full DGD}. Notably, the budget cost is not affected by the noise parameters \(\sigg\) and \(\sigc\).

\subsubsection{Gradient-noise-dominated regime of Full DGD}
\label{DGDr3}

When $\Es$ is moderately small, term $\text{(c)}$ becomes the smallest in \eqref{etaopt}, yielding the optimal learning stepsize as
\begin{align}
\label{ettagraddgd}
\etas
=
\frac{2 L}{\kappa+1}\Big(\frac{\Es}{4\M\sigg}\Big)^2.
\end{align}
By enforcing $\text{(a)},\text{(e)},\text{(f)}\ge \text{(c)}$ in \eqref{etaopt}, 
operating in this regime requires
\begin{subequations}\label{cond3}
\begin{empheq}[left=\empheqlbrace]{align}
\label{cond3.a}&\Es\ge (4\M+1)\sqrt{\DNR/\GCR},\\
\label{cond3.b}&\Es\le \frac{8\M^2}{4\M+1}\frac{\sigg^2}{\mu L}\frac{1}{\sqrt{\DNR}},\\
\label{cond3.c}&\Es\le \dfrac{2\sqrt{2}\M\sigg}{\sqrt{\mu L}},
\end{empheq}
\end{subequations}
which is feasible  provided that
\begin{align}
\label{dnrcondforgrad}
\phi\frac{\mu L\,\DNR}{\sigg^2}
\le
\sqrt{\GCR}\min\{\sqrt{\GCR},1\},
\end{align}
with $\phi$ defined in \eqref{PHI}.
Replacing $\etas$ into \eqref{budgbo2}, the stage-wise
budget cost becomes
\begin{align}
\label{deltagradnoise}
\bd \cdot \deltas
\leq
\bd\cdot
\left\lceil
8\M^2\ln(4\M)\frac{\kappa+1}{\mu L}
\Big(\frac{\sigg}{\Es}\Big)^2
\right\rceil.
\end{align}
In this regime, gradient noise is the dominant bottleneck, whereas the effects of objective heterogeneity, communication noise, and network connectivity are comparatively small. We therefore refer to this regime as \emph{gradient-noise-dominated}. As a result, the budget cost depends primarily on \(\sigg^2\) and is independent of \(\sigc\) and the DNR.

\subsubsection{Communication-noise-dominated regime of Full DGD}
\label{DGDr4}

When $\Es$ is small, term $\text{(a)}$ becomes the smallest in \eqref{etaopt},
yielding 
the optimal learning stepsize as
\begin{align}
\label{optetacomm}
\etas
=
\frac{1}{L+\mu}
\left(
\frac{\sqrt{\kappa}}{\kappa+1}
\frac{\Es^2}{\sqrt{2}\M(4\M+1)\sigc\sqrt{\DNR}}
\right)^2.
\end{align}
By enforcing $\text{(c)},\text{(e)},\text{(f)}\ge \text{(a)}$,
the conditions required to operate in this regime are
\begin{subequations}\label{cond4}
\begin{empheq}[left=\empheqlbrace]{align}
\label{cond4.a}&\Es\le (4\M+1)\sqrt{\DNR/\GCR},\\
\label{cond4.b}&\Es\le \sqrt[3]{\frac{(\kappa+1)^2}{\kappa}\,2\M^2(4\M+1)\sigc^2\sqrt{\DNR}},\\
\label{cond4.c}&\Es\le \dfrac{\sqrt{\kappa+1}}{\sqrt[4]{\kappa}}\sqrt{\M(4\M+1)\sigc}\,\sqrt[4]{2\DNR}.
\end{empheq}
\end{subequations}

Replacing $\etas$ into \eqref{budgbo2}, we obtain
the stage-wise budget cost as
\begin{align}
\label{deltacomm}
\bd{\cdot}\deltas
{\leq}
\bd  \left\lceil
\twocol{\!}2
\M^2(4\M{+}1)^2\ln(4\M)\frac{(\kappa{+}1)^3}{\kappa}
\frac{\sigc^2\DNR}{\Es^4}
\right\rceil\twocol{\!}.\twocol{\!\!\!}
\end{align}
In this regime, the gradient-noise constraint becomes inactive, and the performance is instead limited by the interaction between communication noise and the DNR, which captures the combined effects of objective heterogeneity and network connectivity. Indeed, communication is the mechanism by which agents reconcile heterogeneous local objectives, so communication noise directly impairs this process. We therefore refer to it as the \emph{communication-noise-dominated regime} of Full DGD.
Accordingly, the budget cost is governed by \(\sigc^2\) and the DNR, while being independent of the gradient-noise level \(\sigg\). 

\section{Multi-Stage Budget Complexity}\label{complexity}
The previous section characterized the budget cost and stepsize design of a generic stage under both Local-SGD and Full-DGD operations. We now leverage these stage-wise results to characterize the minimum communication-computation budget required to attain a prescribed accuracy by accumulating the contributions of the bottleneck-dominated regimes traversed during optimization.

For illustrative purposes, we focus on the case where the initial error bound satisfies \eqref{E0init} and the target accuracy $\err$ satisfies \eqref{cond4.a}-\eqref{cond4.c}. More generally, the analysis extends directly to any initial error and target accuracy by simply omitting operating regimes that are never encountered.

Consider stage \(s\), starting at iteration \(k_s\) with error bound
\(
\Es=\E{0}\M^{-s},
\)
where each preceding stage has reduced the error by a factor $\M$. The first question is: given $\Es$ and the target error reduction by a factor $\M$,
which operating mode is most budget-efficient, and which regime can actually be encountered? The answer is given by the following lemma.
\begin{lemma}\label{lembudget}
If \(\Es>(3\M+1)\Eloc\), then Local SGD is budget-optimal. 
Otherwise, Full DGD is budget-optimal.
Moreover, the initialization-dominated regime of Full DGD is never encountered.
\end{lemma}
\begin{proof}
See Appendix~\ref{proofoflembudget}.
\end{proof}

The intuition is that objective heterogeneity becomes relevant only when the error level is comparable to \(\Eloc\). When \(\Es\gg\Eloc\), the current error dominates the effect of heterogeneity, so communication provides little benefit and the lower-cost Local-SGD iterations are sufficient. As the error decreases toward \(\Eloc\), heterogeneity becomes the limiting factor, and communication is required to reconcile the local objectives and enable further progress toward the global optimum. At that point, Full DGD becomes the budget-optimal operating mode.

Moreover, this transition implies that the initialization-dominated regime of Full DGD is never encountered. By the time the algorithm enters the Full-DGD phase, the initialization error has already ceased to be the dominant limitation. Consequently, the Full-DGD phase can only traverse the DNR-, gradient-noise-, and communication-noise-dominated regimes.

Lemma~\ref{lembudget} therefore partitions the optimization process into two consecutive operating phases. Starting from  \(\E{0}\) satisfying \eqref{E0init}, the algorithm first operates in the Local-SGD phase for \(S_{\mathrm{gd}}\) stages, until
\(
\frac{\E{0}}{\M^{S_{\mathrm{gd}}}}
\le
(3\M+1)\Eloc .
\)
It then permanently switches to the Full-DGD phase, which continues for the remaining \(S_{\mathrm{dgd}}\) stages, until the target accuracy \(\err\) is attained.

The single-stage analysis of \secref{singlestage} further revealed that each operating phase may consist of multiple bottleneck-dominated regimes. Within a given regime, the stage-wise budget follows a common scaling law, which remains valid until the error reaches the corresponding regime-transition threshold. The remaining task is therefore to accumulate the budget incurred across successive stages belonging to each regime. This is accomplished by Theorem~\ref{thm:cumbud} (Appendix~\ref{app:cumbud}), which converts the stage-wise budget scaling into the total budget required to traverse a regime. The following corollary summarizes the qualitative implications of Theorem~\ref{thm:cumbud}. 
\begin{corollary}
\label{cor:cumbud}
Consider a bottleneck-dominated regime spanning stages
\(r=s,s+1,\ldots\),
in which the stage lengths satisfy
$
\delta_r=\mathcal O(\E{r}^{-n}),
$
for some exponent \(n\ge0\). Then, the cumulative budget incurred while traversing this regime, until the RMSE falls below a target level \(\errt\), satisfies
\[
B(\errt)=
\begin{cases}
\mathcal O(\errt^{-n}), & n>0,\\[1mm]
\mathcal O(\ln(\Es/\errt)), & n=0,
\end{cases}
\]
while the corresponding RMSE satisfies
\[
E_t=
\begin{cases}
\mathcal O(t^{-1/n}), & n>0,\\
\mathcal O(e^{-ct}), & n=0,\text{ for some \(c>0\).}
\end{cases}
\]
\end{corollary}
Corollary~\ref{cor:cumbud} highlights the qualitative relationship between the stage-wise and global optimization problems. In particular, it shows how the stage-length scaling exponent $n$ determines the  budget complexity and convergence rate of each bottleneck-dominated regime. The exact quantitative analysis presented next is based on Theorem~\ref{thm:cumbud}, which retains the explicit dependence on the DNR, GCR, target accuracy, and other problem parameters. 

We now characterize the total budget of the Local-SGD and Full-DGD phases. Each phase is analyzed by partitioning it into its constituent bottleneck-dominated regimes, applying Theorem~\ref{thm:cumbud} to each regime, and summing the resulting budget contributions. For each regime, the target accuracy \(\errt\) corresponds to the error threshold at which the subsequent regime becomes active, while the stage-length scaling law, and hence the exponent \(n\), is determined by the corresponding bottleneck. For the final regime, \(\errt\) coincides with the overall target accuracy \(\err\).
In particular, applying this methodology to the  Full-DGD phase identifies the four DNR--GCR settings 
shown in Fig.~\ref{fig:regions},
and characterizes the corresponding budget-complexity.
For concreteness, we also provide explicit expressions for the representative choice \(\M=\sqrt{2}\).

\subsection{Budget complexity of the Local-SGD phase}

Recall that Local SGD remains budget-optimal as long as the error exceeds the threshold \((3\M+1)\Eloc\) (Lemma~\ref{lembudget}). 
 By \eqref{Sdef}, the number of Local SGD stages is
\begin{align}
\label{Sgd}
S_{\mathrm{gd}}
\triangleq
\Big\lceil
\log_\M\Big(
\frac{\E{0}}{(3\M+1)\Eloc}
\Big)
\Big\rceil.
\end{align}
Since \(\E{0}\) satisfies \eqref{E0init}, the algorithm initially operates in the initialization-dominated regime of Local SGD (\secref{SGDr1}). Depending on the relative magnitude of the gradient-noise level and the heterogeneity threshold \((3\M+1)\Eloc\), the optimization trajectory may subsequently enter the gradient-noise-dominated regime (\secref{SGDr2}) before transitioning to the Full-DGD phase. 
We then obtain two settings.

\subsubsection{Small-gradient-noise}
This setting is characterized by
\begin{align*}
\boxed{
\sigg{\leq}\Big(1+\frac{1}{3\M}\Big)\sqrt{\mu L}\Eloc,
}
\end{align*}
in which the gradient-noise level is small relative to the objective heterogeneity. The gradient-noise-dominated regime of Local SGD thus becomes infeasible. Consequently, throughout the entire Local-SGD phase the algorithm operates in the initialization-dominated regime (\secref{SGDr1}) with a constant learning stepsize. 

The total budget incurred over the resulting \(S_{\mathrm{gd}}\) stages can be obtained by applying Theorem~\ref{thm:cumbud} to the stage complexity bound \eqref{locsgdinit}. Specifically,
comparing \eqref{locsgdinit} with $\delta_r \le
\lceil\nu\;\E{r}^{\,-n}\rceil$ in \eqref{fghjfsd}, we set
$n=0$, $\nu=\ln(\sqrt{3\M})(\kappa+1)$, and the target error as \(\errt=(3\M+1)\Eloc\),
thus bounding the budget cost of the Local SGD phase via \eqref{Kbound} as
\[
\Bg\lesssim 7.1\;\bg\;\kappa\ln\Big(0.27\frac{\E{0}}{\Eloc}\Big)
\]
for the case \(\M=\sqrt{2}\),
after using \(\lceil x\rceil\leq 1+x\) to bound $S_{\mathrm{gd}}$ in \eqref{Sgd}.
Furthermore, $E_{t}$ can be bound via \eqref{Ebound} with \(\alpha\approx 1.6\) (cf. \eqref{EellboundSGD}), yielding
\begin{align}
\label{ngd}
E_{t}
\lesssim
1.6\E{0}\;
\exp\Big(
-0.12\frac{t+1}{\kappa}
\Big),
\end{align}
for $t=0,\ldots,k_{S_{\mathrm{gd}}}-1$.

Thus, throughout the Local-SGD phase, the error decays geometrically, analogous to noiseless gradient descent on a strongly convex objective. As a result, the associated budget cost grows only logarithmically with the ratio \(\E{0}/\Eloc\). The phase ends when the error reaches the heterogeneity threshold \((3\M+1)\Eloc\). At this point, objective heterogeneity becomes an active bottleneck and communication becomes necessary for further error reduction.

\subsubsection{Large-gradient-noise}
This setting is characterized by
\begin{align*}
\boxed{
\sigg{>}\Big(1+\frac{1}{3\M}\Big)\sqrt{\mu L}\Eloc.
}
\end{align*}
Thus, the algorithm first operates in the initialization-dominated regime for \(S_\init\) stages, until
the error falls below the threshold \(\errt=\frac{3\M\sigg}{\sqrt{\mu L}}\) given by \eqref{E0init}.
Similar to the previous setting, but
 with the target error \(\errt\) in place of $(3\M+1)\Eloc$, we bound
 the total budget cost incurred during the initialization-dominated regime as
\[
\bg\cdot \Big(S_\init+\log_\M(3\M)\kappa\ln\Big(\frac{\sqrt{\mu L}\E{0}}{3\sigg}\Big)
\Big),
\]
during which the RMSE exhibits the geometric decay in \eqref{ngd}.
The algorithm then enters the gradient-noise-dominated regime of Local SGD for the subsequent \(S_\gnd\) stages, until the error falls below the threshold $\errt=(3\M+1)\Eloc$.
Applying Theorem~\ref{thm:cumbud} to the stage complexity bound \eqref{deltagradsgd} in \secref{SGDr2}, with
$n=2$ and
$\nu=
9\M^2\ln(3\M)\frac{\sigg^2}{\mu^2}$, we bound the budget cost accumulated during the gradient-noise-dominated regime as
\[
\bg\;\left(
S_{\gnd}
+
\frac{9\M^4\ln(3\M)}{(\M^{2}-1)(3\M+1)^2}\Big(\frac{\sigg}{\mu\Eloc}\Big)^2
\right),
\]
during which the RMSE decays as ($\M=\sqrt{2}$ and \(\alpha\approx 1.6\) from \eqref{EellboundSGD})
\[
E_{t}
\lesssim
16.3\frac{\sigg}{\mu}\frac{1}{\sqrt{t-k_{S_\init}+1}},
\qquad
t=k_{S_\init},\ldots,k_{S_{\mathrm{gd}}}-1.
\]
Thus, after an initial geometric phase, the convergence behavior 
exhibits the characteristic \(\mathcal O(1/\sqrt{t})\) decay 
of stochastic gradient descent on strongly convex objectives
\cite{7405263}.
In this regime, gradient noise constitutes the dominant bottleneck, and both the error and the budget cost scale with the gradient-noise variance.

Combining the contributions of the two regimes yields
the total budget cost of Local SGD as
\begin{align*}
\Bg\lesssim 
\bg\cdot &\Big[
2.9\ln\Big(0.27\frac{\E{0}}{\Eloc}\Big)
+4.2\kappa\ln\Big(\frac{\sqrt{\mu L}\E{0}}{3\sigg}\Big)
\twocol{\\}&\onecol{\hspace{-25mm}}
+1.9\Big(\frac{\sigg}{\mu\Eloc}\Big)^2
\Big]
\end{align*}
for \(\M=\sqrt{2}\),
where we used \(S_{\mathrm{gd}}=S_\init+S_\gnd\) together with \(\lceil x\rceil\leq 1+x\) in \eqref{Sgd}.

The Local-SGD phase terminates once the error falls below \((3\M+1)\Eloc\). Beyond this point, objective heterogeneity becomes an active bottleneck, and local computation alone can no longer achieve the desired stage-wise error reduction. Communication is therefore required to reconcile objective heterogeneity, motivating the transition to the Full-DGD phase analyzed next.

\subsection{Budget complexity of the Full-DGD phase}
\label{labelcompl}
After the \(S_{\mathrm{gd}}\) stages of Local SGD, the error satisfies
\(
E_{k_{S_{\mathrm{gd}}}}
\le
\E{S_{\mathrm{gd}}}
\le
(3\M+1)\Eloc.
\)
The algorithm then switches to Full DGD and continues until the target accuracy \(\err\) is reached,
while traversing the DNR-, gradient-noise-, and communication-noise-dominated regimes (the initialization-dominated regime is unfeasible, as shown in Lemma \ref{lembudget}). The resulting behavior depends on the relative magnitudes of these bottlenecks, giving rise to four settings that partition the DNR--GCR parameter space into the four regions shown in Fig.~\ref{fig:regions}, each associated with a distinct bottleneck sequence and budget-complexity characterization, analyzed next.

We let \(S_\dnr\), \(S_\gnd\), and \(S_\cnd\) denote the numbers of stages spent in each regime, in order, equal to zero if the corresponding regime is infeasible or never reached. 
Since each stage reduces the error bound by a factor \(\M\), the total number of stages of the Full-DGD phase,
\(
S_{\mathrm{dgd}}=S_\dnr+S_\gnd+S_\cnd,
\)
satisfies (see \eqref{Sdef})
\begin{align}
\label{Sdgd}
S_{\mathrm{dgd}}
\triangleq
\left\lceil
\log_\M\left(\E{S_{\mathrm{gd}}}/\err
\right)
\right\rceil
{\le}
\log_\M\left(
\frac{\M(3\M{+}1)\Eloc}{\err}
\right),
\end{align}
where we used \(\E{S_{\mathrm{gd}}}\le (3\M+1)\Eloc\) and \(\lceil x\rceil\le 1+x\).

\begin{figure}[t]
\centering
\includegraphics[width=\onecol{.7}\linewidth]{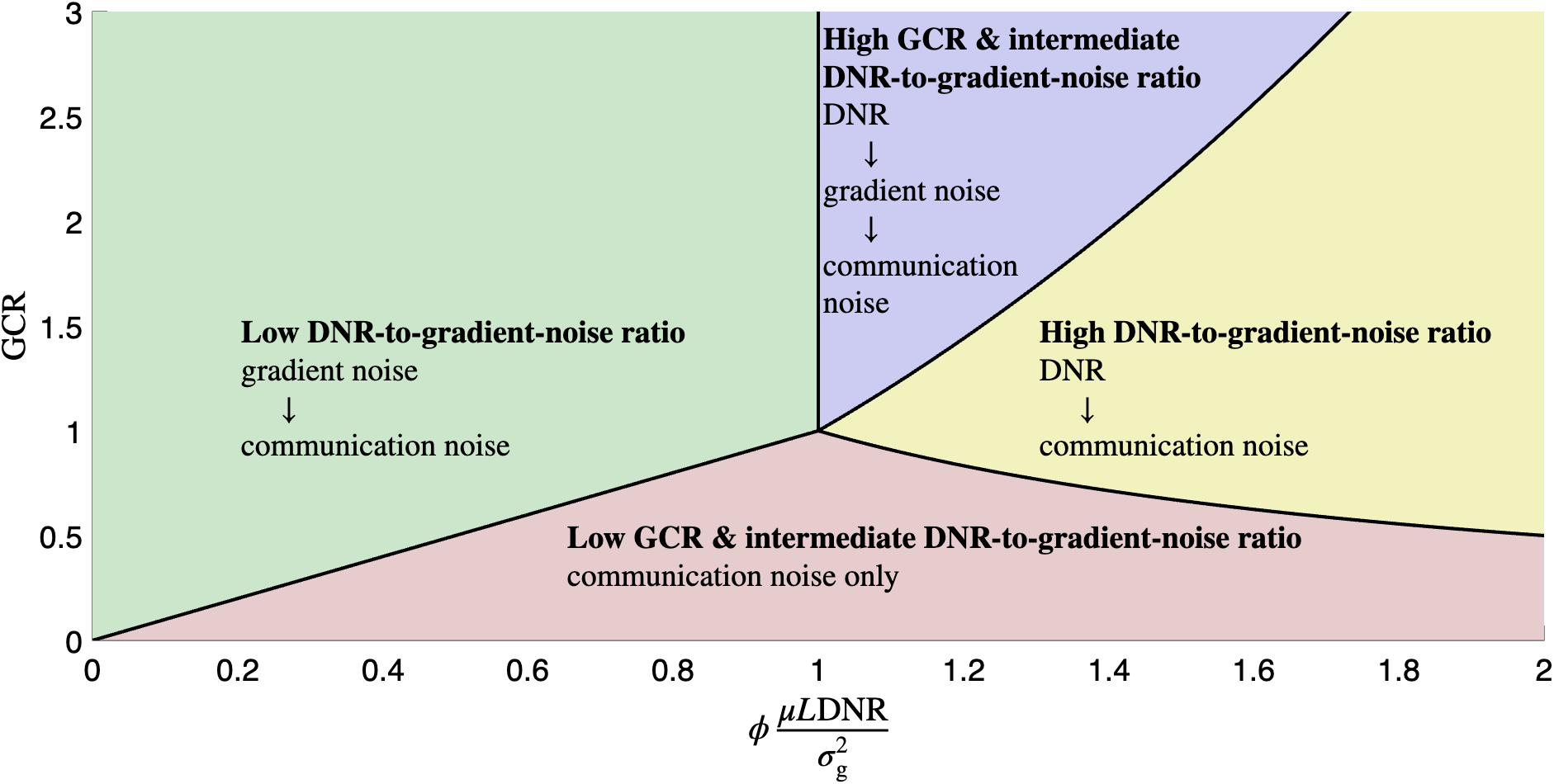}
\caption{
Classification of the four settings governing the Full-DGD phase as a function of the DNR-to-gradient-noise ratio and GCR. Each region corresponds to a distinct sequence of bottlenecks encountered during optimization, indicated below the setting label. 
The boundaries are determined by the setting conditions derived in \secref{labelcompl}.
}
\label{fig:regions}
\end{figure}

The setting conditions presented below characterize the bottleneck structure of the Full-DGD phase in isolation. Since Full DGD is entered only after the Local-SGD phase has already reduced the error below \((3\M+1)\Eloc\), some regimes may become infeasible. For clarity, we discuss these additional feasibility conditions separately as a remark for each setting.

\subsubsection{Low GCR, intermediate DNR-to-gradient-noise ratio}
This  setting is characterized by
\begin{align*}
\boxed{
\GCR\leq
\phi
\frac{\mu L\DNR}{\sigg^2}
\leq \frac{1}{\GCR},
\qquad
\GCR<1,
}
\end{align*}
with $\phi$ defined in \eqref{PHI}.
In other words, the communication noise is large relative to the gradient noise (low GCR), while the DNR is of intermediate magnitude relative to the gradient-noise level.
Under these conditions, it follows from \eqref{condx} and \eqref{dnrcondforgrad} that both the DNR-dominated and gradient-noise-dominated regimes are infeasible, i.e.,
\(
S_\dnr=S_\gnd=0.
\)
Consequently, the Full-DGD phase operates entirely in the communication-noise-dominated regime.

Applying Theorem~\ref{thm:cumbud} to the stage complexity bound \eqref{deltacomm} in \secref{DGDr4}, with
\begin{align}
\label{nnufulld1}
n=4,
\quad
\nu=
2\M^2(4\M+1)^2\ln(4\M)\frac{(\kappa+1)^3}{\kappa}\sigc^2\DNR,
\end{align}
and target error \(\errt=\err\), bounds the total budget $\Bd$ of the Full-DGD phase as
\begin{align}
\label{budcostcommdom}
\Bd
\leq
\bd\left(
S_\cnd
+
\phi_\cnd
\kappa^2
\frac{\sigc^2\DNR}{\err^{4}}
\right),
\end{align}
where
\(
\phi_\cnd
\triangleq
\frac{16\M^6(4\M+1)^2\ln(4\M)}
     {\M^{4}-1}
\).
Specializing to \(\M=\sqrt{2}\) and using 
\eqref{Sdgd}, we obtain
\[
\Bd
\lesssim
\bd\left(
5.8+2.9\ln\left(\frac{\Eloc}{\err}\right)
+
3276
\kappa^2
\frac{\sigc^2\DNR}{\err^{4}}
\right).
\]
Applying \eqref{Ebound} with \(n\) and $\nu$ as in \eqref{nnufulld1} and \(\alpha\approx 1.65\) (cf.~\eqref{dgdalpha}) yields
\begin{align}
\label{errdec}
E_{t}
\lesssim
15
\sqrt{\kappa}
\sqrt[4]{
\frac{\sigc^2\DNR}
{t-k_{s}+1}
},
\end{align}
for all \(t\ge k_{s}\) and $s=S_{\mathrm{gd}}$.
Thus, the error decays as \(\mathcal O(1/\sqrt[4]{t})\), matching the characteristic behavior of DGD with communication noise on strongly convex objectives \cite{10680589}. In this case, the dominant contribution to the budget scales as
\(
\kappa^2
\frac{\sigc^2\DNR}{\err^{4}},
\)
dependent on both the communication-noise level, the DNR and the target accuracy,
while remaining independent of the gradient-noise level \(\sigg\).
 Notably, the quartic dependence on \(1/\err\) implies that attaining high accuracy becomes increasingly expensive once communication noise becomes the dominant bottleneck.

\subsubsection{High DNR-to-gradient-noise ratio}
\label{highdnrset}
This setting is characterized by
\[
\boxed{
\phi
\frac{\mu L\DNR}{\sigg^2}
\ge
\frac{\sqrt{\GCR}}{\min\{\GCR^{3/2},1\}},
}
\]
with $\phi$ defined in \eqref{PHI}.
In other words, the DNR is large relative to the gradient-noise level. Under this condition, the gradient-noise-dominated regime is unfeasible.\footnote{This can be seen by direct inspection of \eqref{condx} and \eqref{dnrcondforgrad}, for the cases $\GCR\geq 1$ and $\GCR<1$ separately} Consequently, the Full-DGD phase consists of a DNR-dominated regime, followed by a communication-noise-dominated regime until the target accuracy \(\err\) is attained.
Intuitively, objective heterogeneity and limited network connectivity initially constitute the dominant bottleneck, making the DNR the primary factor governing the convergence behavior. As the error decreases, the effect of communication noise eventually becomes dominant, causing the algorithm to transition to the communication-noise-dominated regime.

Specifically,  the network remains in the DNR-dominated regime until the error reaches
\[
\errt
=
\sqrt[3]{
\frac{(\kappa+1)^2}{\kappa}\,
2\M^2(4\M+1)
\sigc^2
\sqrt{\DNR}
},
\]
as given by \eqref{cond2.a}.\footnote{By direct inspection, \eqref{cond2.a} is stricter than \eqref{cond2.b} if and only if
\(
\phi
\frac{\mu L\DNR}{\sigg^2}
\ge
\sqrt{\GCR}
     ,
\)
which is implied by the defining condition of this setting.} 
 Comparing \eqref{deltaDNR} in \secref{DGDr2} with $\delta_r \le
\lceil\nu\;\E{r}^{\,-n}\rceil$ in \eqref{fghjfsd}, we set
\begin{align}
\label{nnufulld2}
n=1,
\qquad
\nu=\ln(4\M)(4\M+1)(\kappa+1)\sqrt{\DNR}.
\end{align}
Substituting these quantities into \eqref{Kbound}, we bound
 the total budget accumulated until the error falls below \(\errt\) as
\[
\bd\!\left(
S_\dnr
+
\phi_\dnr
\sqrt[3]{\kappa^2\frac{\DNR}{\sigc^2}}
\right),
\]
where
\(
\phi_\dnr
\triangleq
\frac{\ln(4\M)}{\M-1}
\sqrt[3]{\M(4\M+1)^2}.
\)
Furthermore, specializing \eqref{Ebound} to $n$ and $\nu$ as in \eqref{nnufulld2} and \(\alpha\approx 1.65\) (cf.~\eqref{dgdalpha}),
  yields, for
\(
k_{s}
\le
t
\le
k_{s+S_\dnr}
\), $s=S_{\mathrm{gd}}$, and $\M=\sqrt{2}$,
\begin{align}
\label{errdec2}
E_{t}
\;\lesssim\;
260\,
\frac{\kappa\sqrt{\DNR}}
     {t-k_{s}+1}.
\end{align}
Thus, the \(\mathcal O(1/t)\) error decay  matches the characteristic behavior of noiseless DGD on strongly convex objectives \cite{10680589}. In this regime, objective heterogeneity and limited network connectivity constitute the dominant bottleneck. Consequently, both the error and the dominant contribution to the budget cost scale with DNR, while remaining independent of the gradient-noise level \(\sigg\). Notably, 
larger communication noise leads to an earlier transition to the communication-noise-dominated regime (higher $\errt$), hence
the budget cost decreases with the communication-noise level \(\sigc\).

At this point, the network enters the communication-noise-dominated regime of Full DGD, whose budget cost and error decay are given by \eqref{budcostcommdom} and \eqref{errdec} (with $s=S_{\mathrm{gd}}+S_\dnr$), respectively. Combining the contributions of the DNR-dominated and communication-noise-dominated regimes yields the following bound on the total budget cost for  \(\M=\sqrt{2}\):
\[
\Bd\lesssim\bd\!\left(
5.8 + 2.9 \ln\!\left(\frac{\Eloc}{\err}\right)
+16.6\sqrt[3]{\kappa^2\frac{\DNR}{\sigc^2}}\twocol{\right.\]\[\left.}
+ 3276\kappa^2\frac{\sigc^2\DNR}{\err^{4}}
\right).
\]
Notably, the budget cost increases with the DNR, reflecting the fact that larger objective heterogeneity and weaker network connectivity make it more difficult for the agents to reconcile their local objectives and reach consensus on the global solution.

\begin{remark}[Interaction with the Local-SGD phase]
Since the initialization of Full DGD satisfies \(\Es\le (3\M+1)\Eloc\), the 
DNR-dominated regime is skipped whenever $(3\M+1)\Eloc$ violates
 the lower threshold condition \eqref{cond2.a} required to operate in the DNR-dominated regime.
In this case, the corresponding budget term is omitted from the overall budget expression.
\end{remark}

\subsubsection{Low DNR-to-gradient-noise ratio}
This setting is characterized by
\begin{align*}
\boxed{
\phi\frac{\mu L\DNR}{\sigg^2}
\leq
\min\{\GCR,1\}
,}
\end{align*}
where \(\phi\) is defined in \eqref{PHI}.
In other words, the DNR is small relative to the gradient-noise level. Under this condition, the DNR-dominated regime is unfeasible (seen by direct inspection of \eqref{condx} and \eqref{dnrcondforgrad}). Consequently, the Full-DGD phase consists of a gradient-noise-dominated regime, followed by a communication-noise-dominated regime that persists until the target accuracy \(\err\) is attained.

Intuitively, gradient noise is initially the dominant bottleneck. As the error decreases, the gradient-noise bottleneck is eventually overcome and communication noise becomes the dominant bottleneck, triggering the transition to the communication-noise-dominated regime.

Thus, the network operates in the gradient-noise-dominated regime until reaching the error 
\[
\errt
=
(4\M+1)\sqrt{\DNR/\GCR}
\]
given by \eqref{cond3.a}.
 Comparing \eqref{deltagradnoise} in \secref{DGDr3} with $\delta_r \le
\lceil\nu\;\E{r}^{\,-n}\rceil$ in \eqref{fghjfsd}, we set 
\begin{align}
\label{nnufulld3}
n=2,\qquad
\nu=8\M^2\ln(4\M)\frac{\kappa+1}{\mu L}\sigg^2,
\end{align}
yielding the total budget accumulated until the error falls below \(\errt\) as
\begin{align}
\label{budgrad}
\bd\!\left(
S_\gnd+\phi_\gnd\frac{\sigg^2\;\GCR}{\mu^2\;\DNR}
\right),
\end{align}
where we 
defined the scaling factor \(\phi_\gnd=\frac{16\M^4\ln(4\M)}{(\M^{2}-1)(4\M+1)^2}\).
Furthermore, specializing \eqref{Ebound} to $n$ nd $\nu$ as in \eqref{nnufulld3},
 \(\alpha\approx 1.65\) (cf.~\eqref{dgdalpha}),
 and $\M=\sqrt{2}$,
 yields
\begin{align}
\label{budgraderrdec}
E_{t}
\lesssim
24.6
\frac{\sigg}{\mu}
\frac{1}{\sqrt{t-k_s+1}},
\end{align}
for all iterations belonging to the gradient-noise-dominated regime and $s=S_{\mathrm{gd}}$.
Thus, the error decays as \(\mathcal O(1/\sqrt{t})\), matching the characteristic behavior of stochastic DGD on strongly convex objectives \cite{7405263}. 
Notably, both the error and the dominant contribution to the budget cost scale with the gradient-noise variance. 

Once the error falls below \(\errt\), the network transitions to the communication-noise-dominated regime
 of Full DGD, whose budget cost and error decay are given by \eqref{budcostcommdom} and \eqref{errdec} (with $s=S_{\mathrm{gd}}+S_\gnd$), respectively. Combining the contributions of the gradient-noise- and communication-noise-dominated regimes, we bound
 the total budget cost of the Full-DGD phase as
  \[
\Bd\lesssim\bd\!\left(
5.8 + 2.9 \ln\!\left(\frac{\Eloc}{\err}\right)
+2.5\frac{\sigg^2}{\mu^2}\frac{\GCR}{\DNR}
\twocol{\right.\]\[\left.}
+ 3276\kappa^2\frac{\sigc^2\DNR}{\err^{4}}
\right),
\]
 for \(\M=\sqrt{2}\).
 The total budget increases with both the gradient-noise and communication-noise levels, consistent with the fact that mitigating stronger noise sources requires more communication and computation resources. 
\begin{remark}[Interaction with the Local-SGD phase]
Since the initialization of Full DGD satisfies \(\Es\le (3\M+1)\Eloc\), the gradient-noise-dominated regime is skipped whenever $(3\M+1)\Eloc$ violates
 the lower threshold condition \eqref{cond3.a} required to operate in this regime.
In this case, the corresponding budget term should be omitted from the overall budget expression.
\end{remark}

\subsubsection{High-GCR, intermediate DNR-to-gradient-noise ratio}
\label{highGCRmedDNR}
This setting is characterized by
\[
\boxed{
1
\le
\phi\frac{\mu L\DNR}{\sigg^2}
\leq
\sqrt{\GCR},
\qquad
\GCR>1,
}
\]
with \(\phi\) defined in \eqref{PHI}.
In other words, communication noise is small relative to gradient noise (high GCR), while the DNR is of intermediate magnitude relative to the gradient-noise level. Under these conditions, both the DNR-dominated and gradient-noise-dominated regimes are feasible (seen by direct inspection of \eqref{condx} and \eqref{dnrcondforgrad}). Consequently, the Full-DGD phase traverses all three bottleneck regimes: it first operates in the DNR-dominated regime, then transitions to the gradient-noise-dominated regime, and finally enters the communication-noise-dominated regime until the target accuracy \(\err\) is attained.


Specifically, the network operates in the DNR-dominated regime until reaching the error  threshold
\[
\errt
=
\frac{8\M^2}{4\M+1}\frac{\sigg^2}{\mu L}\frac{1}{\sqrt{\DNR}},
\]
given by \eqref{cond2.b}.\footnote{Indeed, \eqref{cond2.b} is stricter than \eqref{cond2.a} if and only if
\(
\phi
\frac{\mu L\DNR}{\sigg^2}
\le
\sqrt{\GCR}
     ,
\)
implied by the defining condition of this setting.}
 Comparing \eqref{deltaDNR} in \secref{DGDr2} with $\delta_r \le
\lceil\nu\;\E{r}^{\,-n}\rceil$ in \eqref{fghjfsd}, we set 
$n$ and $\nu$ as in \eqref{nnufulld2}.
Using \eqref{Kbound} in Theorem~\ref{thm:cumbud}
gives the total budget accumulated until the error falls below $\errt$ as
\[
\bd\!\left(
S_\dnr
+\phi_\dnr L^2\frac{\DNR}{\sigg^2}
\right),
\]
where we 
defined the scaling factor $\phi_\dnr=\ln(4\M)\frac{(4\M+1)^2}{4\M(\M-1)}$,
and the error decays as in \eqref{errdec2}.
In this regime, the accumulated budget cost scales linearly with $\DNR$ and inversely with $\sigg^2$, since the error threshold $\errt$ is reached earlier as $\sigg^2$ increases.

At this point, gradient noise becomes the dominant bottleneck, and the network transitions to the gradient-noise-dominated regime of Full DGD. The corresponding budget cost and error decay are given by \eqref{budgrad} and \eqref{budgraderrdec}, respectively.
Finally, once the error threshold \eqref{cond3.a} is reached, 
 the network transitions to the communication-noise-dominated regime of Full DGD, with budget cost and error decay given by \eqref{budcostcommdom} and \eqref{errdec}, respectively.

The total budget is obtained by combining the costs of the three regimes traversed. For \(\M{=}\sqrt{2}\):
  \[
\Bd\lesssim\bd\!\left(
5.8 + 2.9 \ln\!\left(\frac{\Eloc}{\err}\right)
+32.8\frac{L^2}{\sigg^2}\DNR
\twocol{\right.\]\[\left.}
+2.5\frac{\sigg^2}{\mu^2}\frac{\GCR}{\DNR}
+ 3276\kappa^2\frac{\sigc^2\DNR}{\err^{4}}
\right).
\]
The third and fifth terms increase with the DNR, highlighting the fact that larger objective heterogeneity or weaker network connectivity makes it more difficult for the agents to reconcile their local objectives and to overcome communication noise. By contrast, the gradient-noise contribution 
decreases with the DNR, since a larger DNR shortens the duration of the gradient-noise-dominated regime. 

\begin{remark}[Interaction with the Local-SGD phase]
Since the initialization of Full DGD satisfies \(\Es\le (3\M+1)\Eloc\), the 
DNR-dominated regime is skipped whenever $(3\M+1)\Eloc$ violates
 the lower threshold condition \eqref{cond2.b} required to operate in this regime.
Furthermore, the gradient-noise-dominated regime is skipped whenever $(3\M+1)\Eloc$ violates
 the lower threshold condition \eqref{cond3.a} required to operate in it.
In these cases, the corresponding budget terms should be omitted from the overall budget expression.
\end{remark}
\onecol{\vspace{-5mm}}
\section{Numerical Results}\label{numres}

To illustrate the proposed bottleneck-centric framework, we consider a network of \(N=50\) agents arranged in a \(K\)-nearest-neighbor ring graph with \(K=10\). Thus, each agent communicates with its \(K\) nearest neighbors on either side of the ring. The mixing matrix $\W$ is symmetric and doubly stochastic, with weight \(1/(2K)\) assigned to each active link.
Each agent holds a one-dimensional quadratic objective
\[
f_i(x)=\frac{\mu_i}{2}\bigl(x-x_i^{\mathrm{loc}}\bigr)^2,
\qquad x\in\mathbb R,
\]
where the curvature parameters \(\mu_i\) are drawn independently and uniformly from \([\mu,L]\),
with \(\mu=1\) and \(L=4\), while the local minimizers \(x_i^{\mathrm{loc}}\) are drawn independently from \(\mathcal N(0,100)\). All agents are initialized from a common point \(x_{\init}\).

This quadratic setting is particularly insightful because the resulting optimization dynamics admit a closed-form mean-square error, allowing us to isolate and visualize the effects of objective heterogeneity, network connectivity, gradient noise, and communication noise.
The global optimum is
\[
x^\star=
\frac{\sum_{i=1}^{N}\mu_i x_i^{\mathrm{loc}}}
{\sum_{i=1}^{N}\mu_i}.
\]
Communication noise is modeled as additive Gaussian,
\(
\boldsymbol{\epsilon}_{\mc,t}\sim\mathcal N(0,\Sigma_{\mc}),
\)
where $\Sigma_{\mc}=\sigma_q^2\W^2$.
This model can be interpreted as an abstraction of a quantization process, where each agent introduces independent quantization noise of variance \(\sigma_q^2\). 
After aggregation through the mixing matrix, the resulting covariance becomes
$\Sigma_{\mc}=\sigma_q^2\W^2$.
 The corresponding communication-noise level is
\(
\sigc^2=\mathbb E\left[\|\boldsymbol{\epsilon}_{\mc,t}\|_2^2 \mid \mathcal F_t\right]
= \mathrm{trace}(\Sigma_{\mc})
=\sigma_q^2\mathrm{trace}(\W^2).
\)
Similarly, gradient noise is modeled as
\(
\boldsymbol{\epsilon}_{\mg,t}
\sim
\mathcal N(0,\Sigma_{\mg}),
\)
with covariance
\(
\Sigma_{\mg}
=
\frac{\sigg^2}{N}\I_N,
\)
corresponding to independent Gaussian perturbations injected by each agent during gradient computation, each having variance \(\sigg^2/N\).
The gradient and communication noises are assumed independent.

To find the closed-form RMSE, define the bias $D_t$ and error covariance matrix $\Sigma_t$ as
\[
D_t\triangleq \mathbb E[\x_t-\xstar],\ 
\Sigma_t\triangleq\mathbb E\!\left[(\x_t{-}\xstar{-}D_t)(\x_t{-}\xstar{-}D_t)^\top\right].
\]
These are initialized as \(D_0=\x_0-\xstar\) and \(\Sigma_0=\mathbf 0\).
For \(t>0\), \(D_t\) and \(\Sigma_t\) can be computed recursively from the updates \eqref{xtdyn}. Using \(\nabla f(\x_t)=\boldsymbol{\mu}(\x_t-\xloc)\), where \(\boldsymbol{\mu}\) is diagonal with \(i\)th diagonal element \(\mu_i\), we rewrite \eqref{xtdyn} as
\begin{align*}
\x_{t+1}{-}\xstar
&{=}
\mathbf A_t(\x_t{-}\xstar){+}\eta_t\boldsymbol{\mu}(\xloc{-}\xstar)
{+}\mathcal N(\mathbf 0,\gamma_t^2\Sigma_{\mc}{+}\eta_t^2\Sigma_{\mg}),
\end{align*}
where
\(
\mathbf A_t{\triangleq}(1{-}\gamma_t)\I_N{+}\gamma_t \W{-}\eta_t\boldsymbol{\mu}.
\)
Taking expectation yields
\begin{align}
\label{biasev}
D_{t+1}
=
\mathbf A_t D_t+\eta_t\boldsymbol{\mu}(\xloc-\xstar).
\end{align}
Similarly,
\begin{align*}
\x_{t+1}{-}\xstar{-}D_{t+1}
{=}
\mathbf A_t(\x_t{-}\xstar{-}D_t)
{+}\mathcal N(\mathbf 0,\gamma_t^2\Sigma_{\mc}{+}\eta_t^2\Sigma_{\mg}).
\end{align*}
Calculating its covariance yields 
\begin{align}
\label{covev}
\Sigma_{t+1}
=
\mathbf A_t\Sigma_t\mathbf A_t^\top
+\gamma_t^2\Sigma_{\mc}+\eta_t^2\Sigma_{\mg}.
\end{align}
With \(D_t\) and \(\Sigma_t\) thus given, we can compute the RMSE as
\begin{align}
\label{rmselin}
E_t=\sqrt{\mathbb E[\|\x_t-\xstar\|^2]}
=\sqrt{\|D_t\|^2+\mathrm{trace}(\Sigma_t)}.
\end{align}

In addition, it is useful to quantify the contribution of the different bottlenecks to the overall error. To this end, we further decompose the error covariance into the contributions due to communication noise, \(\Sigma_t^{(\mathrm{comm})}\), and that due to gradient noise, \(\Sigma_t^{(\mathrm{grad})}\). These quantities evolve according to
\begin{align*}
\Sigma_{t+1}^{(\mathrm{comm})}
&=
\mathbf A_t\Sigma_t^{(\mathrm{comm})}\mathbf A_t^\top
+\gamma_t^2\Sigma_{\mc},
\twocol{\\}\onecol{\ \ }
\Sigma_{t+1}^{(\mathrm{grad})}
&\onecol{\hspace{-25mm}}=
\mathbf A_t\Sigma_t^{(\mathrm{grad})}\mathbf A_t^\top
+\eta_t^2\Sigma_{\mg},
\end{align*}
with \(\Sigma_0^{(\mathrm{comm})}=\Sigma_0^{(\mathrm{grad})}=\mathbf 0\). By linearity,
\(
\Sigma_t=
\Sigma_t^{(\mathrm{comm})}
+
\Sigma_t^{(\mathrm{grad})}.
\)
Similarly, we decompose the bias into two components: the contribution due to objective heterogeneity and limited network connectivity, and that due to initialization. To isolate the former, consider the noiseless dynamics
\begin{align}
\tilde{\x}_{t+1}
=[(1-\gamma_t)\I_{N}
+\gamma_t\W]\tilde{\x}_t
-\eta_t\nabla f(\tilde{\x}_t),
\end{align}
initialized at the global optimum,
\(\tilde{\x}_0=\xstar.\)
The resulting trajectory captures the bias induced solely by objective heterogeneity and limited connectivity. The remaining bias is due to the fact that the algorithm is initialized at \(\x_{\init}\neq \xstar\), and thus captures the effect of initialization.
Accordingly, we decompose
\(
D_t
=
D_t^{(\DNR)}
+
D_t^{(\mathrm{init})},
\)
where
\[
D_t^{(\DNR)}
\triangleq
\tilde{\x}_t-\xstar,
\qquad
D_t^{(\mathrm{init})}
\triangleq
D_t-D_t^{(\DNR)}.
\]
These quantities satisfy 
$D_0^{(\DNR)}=\mathbf 0,$
$D_0^{(\mathrm{init})}=D_0,$
and for $t\geq 0$,
\begin{align*}
&D_{t+1}^{(\DNR)}
=
\mathbf A_tD_t^{(\DNR)}
+\eta_t\boldsymbol{\mu}(\xloc-\xstar),
\twocol{\\}\onecol{\hspace{-10mm}}
&D_{t+1}^{(\mathrm{init})}
=
\mathbf A_tD_t^{(\mathrm{init})}.
\end{align*}

With these definitions, the overall MSE can be written as
\[
E_t^2
=
\|D_t^{(\init)}+D_t^{(\DNR)}\|^2
+\mathrm{trace}(\Sigma_t^{(\mathrm{comm})})
+\mathrm{trace}(\Sigma_t^{(\mathrm{grad})}).
\]
We then measure the relative contributions of communication and gradient noises to the MSE as
\begin{align*}
\rho_t^{(\mathrm{comm})}
&\triangleq
\frac{\mathrm{trace}(\Sigma_t^{(\mathrm{comm})})}{E_t^2},
\ 
\rho_t^{(\mathrm{grad})}
\triangleq
\frac{\mathrm{trace}(\Sigma_t^{(\mathrm{grad})})}{E_t^2}.
\end{align*}
For the effects of initialization and DNR, the MSE is not additive with respect to their contributions, since the cross product \(D_t^{(\init)\top}D_t^{(\DNR)}\) is generally non-zero. We therefore define their relative contributions as
\begin{align*}
\rho_t^{(\init)}
&\triangleq
\frac{\|D_t^{(\init)}\|^2}{\|D_t^{(\init)}\|^2+\|D_t^{(\DNR)}\|^2}
\cdot
\frac{\|D_t^{(\init)}+D_t^{(\DNR)}\|^2}{E_t^2},
\\
\rho_t^{(\DNR)}
&\triangleq
\frac{\|D_t^{(\DNR)}\|^2}{\|D_t^{(\init)}\|^2+\|D_t^{(\DNR)}\|^2}
\cdot
\frac{\|D_t^{(\init)}+D_t^{(\DNR)}\|^2}{E_t^2}.
\end{align*}
With these definitions, the four relative contributions satisfy
$
\rho_t^{(\init)}+\rho_t^{(\DNR)}+\rho_t^{(\mathrm{grad})}+\rho_t^{(\mathrm{comm})}=1.
$

In the subsequent numerical evaluation, we compare 
the proposed multi-stage schedule
 against several benchmarks. The selected stepsize schedule determines the evolution of the bias and covariance according to \eqref{biasev} and \eqref{covev}, and therefore characterizes the RMSE in \eqref{rmselin}.

\begin{itemize}[leftmargin=*]
\item \emph{Constant stepsizes}: This schedule employs constant learning and consensus stepsizes,
\(
\eta=\frac{1}{L+\mu},
\gamma=\frac{1}{2},
\)
satisfying the convergence conditions in \eqref{origcond}. In the noiseless setting, the resulting iterates converge linearly to a neighborhood of the global optimum, with asymptotic error \eqref{fperror}.

\item \emph{Gradient-noise-aware}: 
This schedule is designed to mitigate the effect of gradient noise in DGD.  
 Communication noise is neglected in its design. It therefore uses a fixed consensus stepsize,
\(
\gamma=\frac{1}{2},
\)
together with a decreasing learning stepsize
\[
\eta_t=
\min\left\{
\frac{1}{\mu t},
\frac{1}{L+\mu}
\right\}.
\]
The \(1/(\mu t)\) decay is motivated by classical stochastic approximation and stochastic gradient methods for strongly convex objectives, see, e.g., \cite{doi:10.1137/070704277}.

\item \emph{Communication-noise-aware}: This schedule is designed to mitigate the combined effects of communication noise, gradient noise, objective heterogeneity, and limited network connectivity. Following \cite{10680589}, the learning and consensus stepsizes are chosen as
\[
\eta_t=
\frac{1}{\mu+L}
\frac{1}{1+\frac{4\mu}{5(\mu+L)}t},
\qquad
\gamma_t=
\frac{0.5}
{\left(1+\frac{4\mu}{5(\mu+L)}t\right)^{3/4}};
\]
 the specific exponents and scaling factors are motivated and theoretically justified in \cite{10680589}.
\end{itemize}

The proposed multi-stage schedule is illustrated in Fig.~\ref{fig:ss} for the case $\M{=}\sqrt{2}$.
During the Local-SGD phase, the consensus stepsize is identically zero, since no communication takes place. As the error decreases, the learning stepsize is progressively reduced from its maximum admissible value \(2/(\mu+L)\). Upon entering the Full-DGD phase, both the learning and consensus stepsizes become functions of the current error level. In particular, the learning stepsize scales linearly with the error in the DNR-dominated regime (\(\etas\propto \Es\)), quadratically in the gradient-noise-dominated regime (\(\etas\propto \Es^2\)), and with the power of four in the communication-noise-dominated regime (\(\etas\propto \Es^4\)), in agreement with the theoretical characterizations derived in \secref{dgdan}.

The consensus stepsize exhibits a different behavior. It remains at its maximum value \(\gammas=1/2\) throughout the DNR-dominated and most of the gradient-noise-dominated regimes, indicating that communication is not yet the limiting factor and that aggressive information mixing is therefore desirable.
However, once communication noise becomes significant, the consensus stepsize is progressively reduced. This reflects the fact that communication becomes increasingly error prone, and a smaller consensus stepsize is required to limit the accumulation of communication noise across iterations.

\begin{figure}[t]
\centering
\includegraphics[width=\twocol{.8}\onecol{.5}\linewidth]{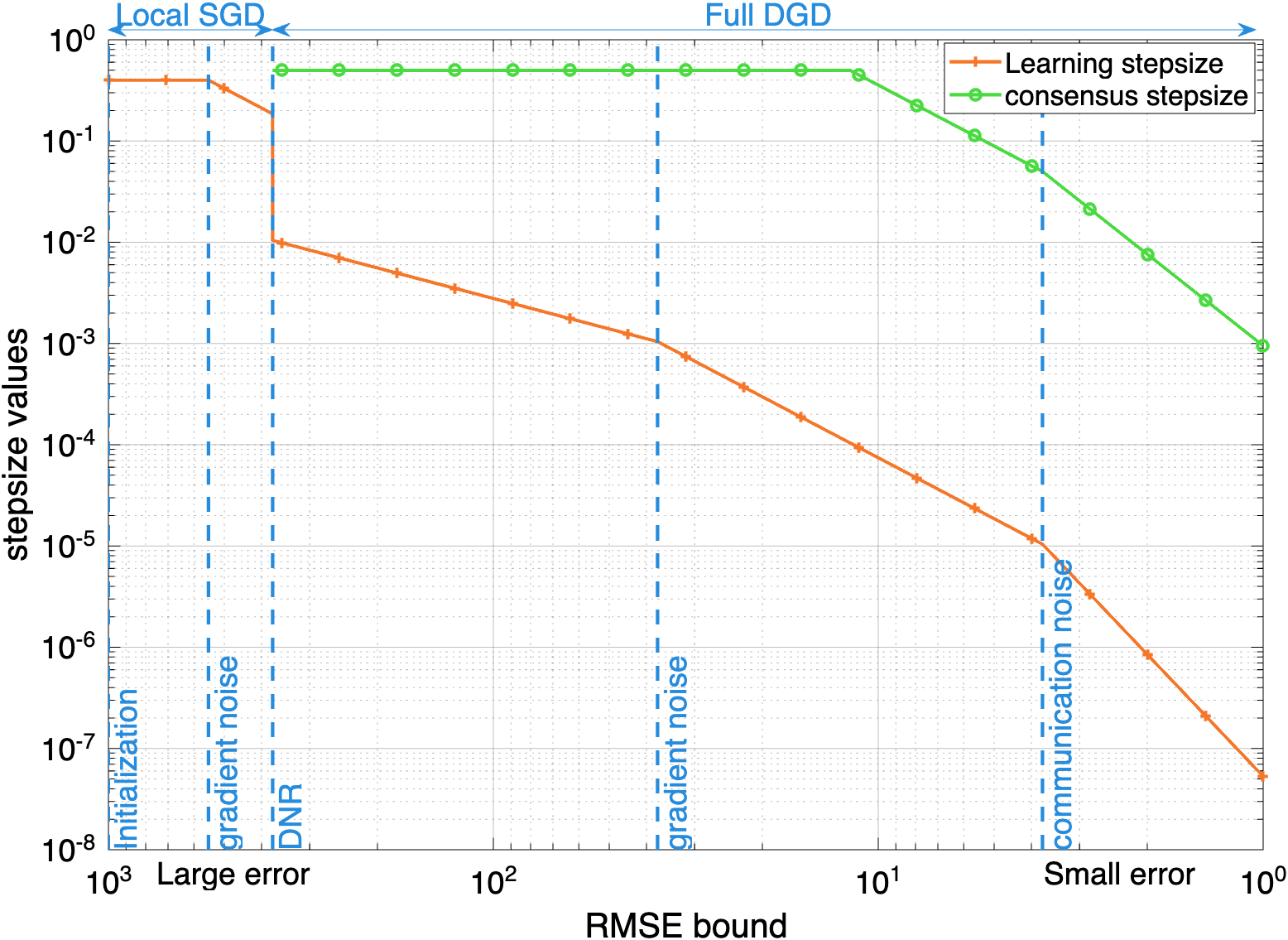}
\caption{Learning and consensus stepsizes, $\etas$ and $\gammas$, as a function of the RMSE bound $\Es$, for the High-GCR, intermediate DNR-to-gradient-noise ratio setting.
The plot shows also the phase and regime transitions under the multi-stage schedule.
}
\label{fig:ss}
\end{figure}

\twocol{\mainfig}

The proposed method transitions from a Local-SGD phase to a Full-DGD phase as dictated by the active bottleneck, whereas all benchmark methods rely exclusively on Full DGD. To facilitate a fair comparison of the resulting optimization dynamics, we adopt the optimistic setting \(\bg=\bd=1\), so that the budget coincides with the total number of iterations. This choice isolates the impact of the different bottlenecks and operating strategies without introducing additional asymmetries between Local-SGD and Full-DGD operations. Consequently, any performance gain observed for the proposed method stems from its bottleneck-aware scheduling strategy rather than from assigning a lower cost to local computation. 

Figure~\ref{fig:set4} shows the RMSE (top row), normalized by the initial error \(\|\x_0-\xstar\|\), together with the relative contributions of the four error components (bottom row) under the proposed multi-stage schedule. The left and right columns correspond to the High-GCR, intermediate DNR-to-gradient-noise ratio setting (\secref{highGCRmedDNR}) and the High DNR-to-gradient-noise ratio setting (\secref{highdnrset}), respectively.
From the RMSE curves, we observe that the proposed multi-stage schedule achieves a sustained error reduction throughout the entire optimization horizon. In contrast, the constant-stepsize schedule saturates due to the heterogeneity-induced bias of DGD. The communication-noise-aware schedule also exhibits a decreasing error, but its convergence is significantly slower because it assumes communication noise to be the dominant bottleneck at all times, resulting in overly conservative stepsizes during the early stages of optimization. The gradient-noise-aware schedule performs well initially, when communication noise contributes little to the overall error. However, once communication noise becomes significant, its fixed consensus stepsize is unable to sufficiently suppress the propagation of communication errors, leading to a deterioration in performance.

\onecol{\mainfig}

The bottom row provides insight into the evolution of the active bottleneck. In the High-GCR, intermediate DNR-to-gradient-noise ratio setting (left column), the optimization trajectory traverses all bottleneck regimes predicted by the theory. The initialization error dominates during the early iterations and rapidly decays. Gradient noise subsequently becomes the dominant source of error during the Local-SGD phase, followed by a DNR-dominated phase after the switch to Full DGD. As the error decreases further, the relative impact of DNR diminishes
and that of gradient-noise remains sustained, until
 communication noise progressively emerges as the dominant bottleneck, eventually accounting for the majority of the MSE.

A similar behavior is observed in the High DNR-to-gradient-noise ratio setting (right column). However, since gradient noise is comparatively weak, the gradient-noise-dominated regimes are skipped. The optimization process therefore transitions directly from the initialization-dominated regime of Local SGD to the DNR-dominated regime of Full DGD, and subsequently to the communication-noise-dominated regime. The corresponding error decomposition reflects this sequence of bottlenecks.

Overall, these results validate the proposed bottleneck-centric framework. The observed transitions closely match the theoretically predicted operating regimes, and the resulting multi-stage schedule adapts the learning and consensus stepsizes to the active bottleneck, improving budget-efficiency compared with the benchmark schemes.
\section{Concluding Remarks} \label{concl}

This paper developed a resource-aware, bottleneck-centric framework for DGD. We showed that objective heterogeneity, network connectivity, gradient noise, and communication noise become dominant at different error scales, giving rise to distinct operating regimes and budget scalings. 
The resulting analysis reveals how the overall budget decomposes into the costs of overcoming successive bottlenecks and provides insight into the tradeoffs among objective heterogeneity, network connectivity, gradient noise, and communication noise. 
More broadly, the proposed framework opens new avenues for resource-aware decentralized optimization, where algorithms adapt to the active bottleneck as the optimization progresses.

\appendices

\section{Upper bound on the initial RMSE}

Suppose all devices initialize the algorithm from the same point
\(
x_{i,0}=x_{\mathrm{init}}, \forall i.
\)
Then, using strong convexity of the global objective $F(x)$, we have
\[
\|\nabla F(x_{\mathrm{init}})\|
=
\|\nabla F(x_{\mathrm{init}})-\nabla F(x^\star)\|
\ge
\mu \|x_{\mathrm{init}}-x^\star\|.
\]
Recalling that $E_0=\sqrt{N}\|x_{\mathrm{init}}-x^\star\|$, this implies
\(
E_0
\le
\frac{\sqrt{N}}{\mu}\|\nabla F(x_{\mathrm{init}})\|.
\)
We continue by using $F(x)=\frac{1}{N}\sum_{i=1}^N f_i(x)$ and the triangle inequality:
\[
E_0
\le\frac{\sqrt{N}}{\mu}\|\nabla F(x_{\mathrm{init}})\|
=
\frac{\sqrt{N}}{\mu}\Big\|\frac{1}{N}\sum_{i=1}^N \nabla f_i(x_{\mathrm{init}})\Big\|\]\[
\le
\frac{1}{\mu\sqrt{N}}\sum_{i=1}^N \|\nabla f_i(x_{\mathrm{init}})\|
\le
\frac{\sqrt{N}}{\mu}\max_i \|\nabla f_i(x_{\mathrm{init}})\|\triangleq
\E{0}.
\]
The bound $\E{0}$ can be computed with low communication overhead via a scalar distributed max-consensus procedure \cite{6259916}. 
 Tighter, albeit more communication-intensive, bounds can be obtained by 
using one of the intermediate bounds in the above sequence of inequalities. Since  $\E{0}$ only serves to initialize the stage-wise recursion, any computable upper bound may be used.

\section{Proof of Lemma \ref{lem:localSGD}}
\label{proofoflem:localSGD}
\begin{proof}
When \(\;0{\leq}\eta{\le}2/(L+\mu)\), 
the mean-square error relative to \(\xloc\) satisfies \cite{10947567}
\[
\mathbb E\big[\|\mathbf x_{t+1}-\xloc\|^2\big]
\le (1-\eta\mu)^2\mathbb E\big[\|\mathbf x_{t}-\xloc\|^2\big]+\eta^2\sigg^2,
\]
where \(\sigg^2\) bounds the conditional second moment of the gradient noise (Assumption \ref{noiseas}). 
Applying this inequality recursively for \(\delta\) local steps
  gives
\begin{align}
\twocol{\nn&\mathbb E\big[\|\mathbf x_{k_s+\delta}-\xloc\|^2\big]\\\nn&}
\onecol{\nn&\mathbb E\big[\|\mathbf x_{k_s+\delta}-\xloc\|^2\big]}
\le (1-\eta\mu)^{2\delta}\mathbb E\big[\|\mathbf x_{k_s}-\xloc\|^2\big]
{+}\eta^2\sigg^2\sum_{r=0}^{\delta-1}(1-\eta\mu)^{2r}\\
&\le (1-\eta\mu)^{2\delta}\mathbb E\big[\|\mathbf x_{k_s}-\xloc\|^2\big]
{+}\frac{\eta^2\sigg^2}{1{-}(1{-}\eta\mu)^{2}}
\twocol{\nn\\} &\onecol{\hspace{-46mm}}
\le (1-\eta\mu)^{2\delta}\mathbb E\big[\|\mathbf x_{k_s}-\xloc\|^2\big]
{+}\frac{\eta(\kappa{+}1)\sigg^2}{2L},
\label{errbo}
\end{align}
where in the last step we used 
\(1-(1-\eta\mu)^2=\eta\mu(2-\eta\mu)\ge 2\eta\mu L/(L+\mu)
\) for \(\eta\le 2/(L+\mu)\).

We are interested in the distance to the global optimum \(\xstar\). By Minkowski's inequality \cite[Lemma 14.10]{florescu2013handbook} and the RMSE definition in \eqref{rmse},
\begin{align*}
E_{k_s+\delta}&=\sqrt{\mathbb E\big[\|\mathbf x_{k_s+\delta}-\xstar\|^2\big]}
\le \sqrt{\mathbb E\big[\|\mathbf x_{k_s+\delta}-\xloc\|^2\big]}+\Eloc\\
&\twocol{\hspace{-0.5cm}}\le (1-\eta\mu)^{\delta}\sqrt{\mathbb E\big[\|\mathbf x_{k_s}-\xloc\|^2\big]}
+\frac{\sqrt{\eta(\kappa+1)}}{\sqrt{2 L}}\sigg+\Eloc,
\end{align*}
where we used \eqref{errbo} in the last inequality.
Next, using 
Minkowski's inequality again to bound
\(\sqrt{\mathbb E\|\mathbf x_{k_s}-\xloc\|^2}\le \sqrt{\mathbb E\|\mathbf x_{k_s}-\xstar\|^2}+\Eloc{=}
E_{k_s}{+}\Eloc\) and $E_{k_s}{\leq}\Es$, we obtain
\begin{align}
\label{errb}
E_{k_s+\delta}
\le& (1-\eta\mu)^{\delta}\Es
+\frac{\sqrt{\eta(\kappa+1)}}{\sqrt{2 L}}\sigg\twocol{\nn\\}&\onecol{\hspace{-28mm}}
+(1+(1-\eta\mu)^{\delta})\Eloc,
\end{align}
which relates the RMSE at time $k_s+\delta$ to that at time $k_s$.
The final result of the lemma is obtained by using $a+b+c\leq 3\max\{a,b,c\}$.
\end{proof}

\section{Proof of Lemma \ref{lem:fullDGD}}
\label{proofoflem:fullDGD}
\begin{proof}
In \cite{10947567}, it was proved that, for $0{<}\eta{\leq}\frac{2-\gamma(1-\lambda_N)}{L+\mu}$ and $\gamma\in(0,1]$, the one-step error dynamic is bounded as
\[
\mathbb E\big[\|\mathbf x_{t+1}{-}\hat{\x}\|^2 \mid \mathcal F_t\big]
{\le} (1{-}\eta\mu)^2\|\mathbf x_t{-}\hat{\x}\|^2
\twocol{\]\[}
+
\mathbb E\big[\|\gamma \boldsymbol{\epsilon}_{\mc,t}
-
\eta \boldsymbol{\epsilon}_{\mg,t}\|^2 \mid \mathcal F_t\big],
\]
where $\hat{\x}$ is the  fixed point (FP) of the noiseless iterate,\footnote{Its existence and uniqueness is guaranteed by the fact that, under the stepsize condition, 
 \eqref{contrct} is a contraction mapping, combined with
  Banach's fixed point Theorem \cite[Th.~9.23]{rudin1953principles}. See \cite{10947567} for further details.} 
  that is, the unique solution of
 \begin{align}
\hat{\x}
=
(1-\gamma)\hat{\x}
+
\gamma (\mathbf{W}\otimes\I_d)\hat{\x}
-
\eta \nabla f(\hat{\x}),
\end{align}
which Full DGD converges to.
Furthermore,
using Minkowski's inequality \cite[Lemma 14.10]{florescu2013handbook},
\[
\sqrt{\mathbb E\big[\|\gamma \boldsymbol{\epsilon}_{\mc,t}
-
\eta \boldsymbol{\epsilon}_{\mg,t}\|^2 \mid \mathcal F_t\big]}
\leq
\gamma\sqrt{\mathbb E\big[\| \boldsymbol{\epsilon}_{\mc,t}\|^2 \mid \mathcal F_t\big]}
\twocol{\]\[}
+\eta\sqrt{\mathbb E\big[ \|\boldsymbol{\epsilon}_{\mg,t}\|^2 \mid \mathcal F_t\big]}
\leq
\gamma\sigc+\eta\sigg
,
\]
where \(\sigg^2\) and \(\sigc^2\) bound the conditional second moments of the gradient and communication noises (Assumption \ref{noiseas}).
Taking total expectation, we then obtain
\begin{align}
\mathbb E\big[\|\mathbf x_{t+1}{-}\hat{\x}\|^2\big]
{\le}& (1{-}\eta\mu)^2\mathbb E\big[\|\mathbf x_{t}{-}\hat{\x}\|^2\big]
{+}
(\gamma\sigc{+}\eta\sigg)^2.
\label{onestepdgd}
\end{align}
Applying this inequality recursively for \(\delta\) steps 
  gives
\begin{align}
\nonumber
&\mathbb E\big[\|\mathbf x_{k_s+\delta}-\hat{\x}\|^2\big]
\twocol{\\}&\onecol{\hspace{-80mm}}\le (1-\eta\mu)^{2\delta}\mathbb E\big[\|\mathbf x_{k_s}-\hat{\x}\|^2\big]
+(\gamma\sigc+\eta\sigg)^2\sum_{r=0}^{\delta-1}(1-\eta\mu)^{2r}
\nonumber
\\&
\nonumber
\le (1-\eta\mu)^{2\delta}\mathbb E\big[\|\mathbf x_{k_s}-\hat{\x}\|^2\big]
+\frac{(\gamma\sigc+\eta\sigg)^2}{1-(1-\eta\mu)^{2}}\\
&\le (1-\eta\mu)^{2\delta}\mathbb E\big[\|\mathbf x_{k_s}-\hat{\x}\|^2\big]
+(\kappa+1)\frac{(\gamma\sigc+\eta\sigg)^2}{2\eta L},
\label{lastineq}
\end{align}
where in the last step we used 
\(1-(1-\eta\mu)^2=\eta\mu(2-\eta\mu)\ge 2\eta\mu L/(L+\mu)
\) since \(\eta\le 2/(L+\mu)\).

The distance to the global optimum follows from
Minkowski's inequality~\cite[Lemma 14.10]{florescu2013handbook},
\twocol{
\begin{align*}
E_{k_s+\delta}&=\sqrt{\mathbb E\big[\|\mathbf x_{k_s+\delta}-\xstar\|^2\big]}
\\&
\le \sqrt{\mathbb E\big[\|\mathbf x_{k_s+\delta}-\hat{\x}\|^2\big]}+\|\hat{\x}-\xstar\|\\
&\le (1-\eta\mu)^{\delta}\sqrt{\mathbb E\big[\|\mathbf x_{k_s}-\hat{\x}\|^2\big]}
+\frac{\sqrt{\eta(\kappa+1)}}{\sqrt{2 L}}\sigg
\\&
\quad+\frac{\gamma}{\sqrt{\eta}}\frac{\sqrt{\kappa+1}}{\sqrt{2 L}}\sigc
+\|\hat{\x}-\xstar\|,
\end{align*}}
\onecol{
\begin{align*}
E_{k_s+\delta}&=\sqrt{\mathbb E\big[\|\mathbf x_{k_s+\delta}-\xstar\|^2\big]}
\le \sqrt{\mathbb E\big[\|\mathbf x_{k_s+\delta}-\hat{\x}\|^2\big]}+\|\hat{\x}-\xstar\|\\
&\le (1-\eta\mu)^{\delta}\sqrt{\mathbb E\big[\|\mathbf x_{k_s}-\hat{\x}\|^2\big]}
+\frac{\sqrt{\eta(L+\mu)}}{\sqrt{2\mu L}}\sigg
+\frac{\gamma}{\sqrt{\eta}}\frac{\sqrt{L+\mu}}{\sqrt{2\mu L}}\sigc
+\|\hat{\x}-\xstar\|,
\end{align*}}
where in the last step we used \eqref{lastineq} and the triangle inequality.
Next, using 
Minkowski's inequality again to bound
\(\sqrt{\mathbb E\|\mathbf x_{k_s}-\hat{\x}\|^2}\le \sqrt{\mathbb E\|\mathbf x_{k_s}-\xstar\|^2}+\|\hat{\x}-\xstar\|\), we obtain
\begin{align}
\label{errdgd}
&E_{k_s+\delta}
\le (1-\eta\mu)^{\delta}E_{k_s}
+\frac{\sqrt{\eta(\kappa+1)}}{\sqrt{2 L}}\sigg\twocol{\nn\\}&\onecol{\hspace{-5mm}}
+\frac{\gamma}{\sqrt{\eta}}\frac{\sqrt{\kappa+1}}{\sqrt{2 L}}\sigc
+(1+(1-\eta\mu)^{\delta})\|\hat{\x}-\xstar\|.
\end{align}
In \cite{10947567}, it was shown that the FP error satisfies
\begin{align}
\label{fperror}
\|\hat{\x}-\xstar\|
\le
\frac{\eta}{\gamma}\kappa\frac{\|\nabla f(\xstar)\|}{(1-\lambda_2)}
=
\frac{\eta}{\gamma}\frac{L+\mu}{2}\sqrt{\DNR}.
\end{align}
(see the definition of $\DNR$ in \eqref{DNR}).  
Substituting \eqref{fperror} into \eqref{errdgd} and using $E_{k_s}\leq \Es$, we obtain
\twocol{
\begin{align}
&E_{k_s+\delta}
\le (1-\eta\mu)^{\delta}\Es
+\frac{\sqrt{\eta(\kappa+1)}}{\sqrt{2 L}}\sigg\nn\\&
+\frac{\gamma}{\sqrt{\eta}}\frac{\sqrt{\kappa+1}}{\sqrt{2 L}}\sigc
+(1+(1-\eta\mu)^{\delta})\frac{\eta}{\gamma}\frac{L+\mu}{2}\sqrt{\DNR},
\label{dgdboun2}
\end{align}}
\onecol{
\begin{align}
&E_{k_s+\delta}
{\le}(1-\eta\mu)^{\delta}\Es
{+}\frac{\sqrt{\eta(L+\mu)}}{\sqrt{2\mu L}}\sigg
{+}\frac{\gamma}{\sqrt{\eta}}\frac{\sqrt{L+\mu}}{\sqrt{2\mu L}}\sigc
+(1{+}(1-\eta\mu)^{\delta})\frac{\eta}{\gamma}\frac{L+\mu}{2}\sqrt{\DNR},
\label{dgdboun2}
\end{align}}
which relates the error after $\delta$ iterations to the error at the start of stage $s$.
The bound in the lemma statement directly follows after using $a+b+c+d\leq4\max\{a,b,c,d\}$.
\end{proof}

\section{Proof of Lemma \ref{lem:fullDGDoptsteps}}
\label{proofoflem:fullDGDoptsteps}
\begin{proof}
We can reformulate \eqref{eq:four_conditions.b}-\eqref{eq:four_conditions.d}
as
\begin{align}
\label{eq2}
\begin{cases}
\eta \le \frac{2 L}{\kappa+1}\Big(\frac{\Es}{4\M\sigg}\Big)^2\triangleq A;\\
\eta\ge\frac{\kappa+1}{2 L} \Big(\frac{4\gamma\M\sigc}{\Es}\Big)^2\triangleq B;\\
\eta\le \frac{2}{L+\mu}\frac{1}{4\M+1}\frac{\gamma\Es}{\sqrt{\DNR} }\triangleq C.
\end{cases}
\end{align}
These conditions must hold together with \eqref{newcond}.
For a given consensus stepsize $0<\gamma\le \frac{1}{2}$, the learning stepsize $\eta$ should be chosen as large as possible. From \eqref{eq2} and \eqref{newcond}, this yields
\[
\eta
=
\min\left\{A,\;C,\;\frac{1}{L+\mu}\right\},
\]
provided that the following feasibility conditions hold:
\begin{align}
\label{feasibility}
B\le A,\qquad
B\le C,\qquad
B\le \frac{1}{L+\mu}.
\end{align}

If any of the conditions in \eqref{feasibility} are violated, then the feasible region for $\eta$ in \eqref{eq2} is empty, resulting in an infeasible stepsize configuration.
Optimizing $\gamma$ therefore amounts to selecting the value of $\gamma$ that maximizes the admissible learning stepsize $\eta$. This leads to the problem
\begin{subequations}
\begin{align}
\max_{\gamma\ge0}\  
&\frac{1}{L{+}\mu}\;\min\Bigg\{
\frac{\mu L}{8\M^2}\Big(\frac{\Es}{\sigg}\Big)^2,
\frac{2\gamma\Es}{(4\M{+}1)\sqrt{\DNR}}
,
1
\Bigg\}
\label{obj}
\\
\text{s.t.}\quad 
&\gamma\leq\frac{L}{\kappa+1}\frac{\Es^2}{8\M^2\sigg\sigc},
\label{c1}\\
&\gamma\leq\frac{\kappa}{(\kappa+1)^2}\frac{\Es^3}{4\M^2(4\M{+}1)\sigc^2\sqrt{\DNR}},
\label{c2}\\
&\gamma\leq
\frac{\sqrt{2\kappa}}{\kappa+1}\frac{\Es}{4\M\sigc},
\label{c3}\\
&\gamma\leq \frac{1}{2},
\label{c4}
\end{align}
\end{subequations}
where the constraints \eqref{c1}-\eqref{c3} correspond to \eqref{feasibility}, and the last one  
corresponds to \eqref{newcond}.
Since the objective \eqref{obj} is non-decreasing in $\gamma$, the optimal $\gamma$ matches the tightest constraint among \eqref{c1}--\eqref{c4}, yielding
\eqref{optgam}.
Substituting $\gammas$ into \eqref{obj} gives the optimal learning stepsize
\[
\etas\triangleq\frac{1}{L+\mu}
\min\Bigg\{
\underbrace{
\Big(\frac{\sqrt{\kappa}}{\kappa+1}\frac{\Es^2}{\sqrt{2}\M(4\M+1)\sigc\sqrt{\DNR}}\Big)^2
}_{\text{(a)}},\]\[
\underbrace{
\frac{L}{\kappa+1}\frac{\Es^3}{4\M^2(4\M+1)\sigg\sigc\sqrt{\DNR}}
}_{\text{(b)}},
\underbrace{
2\mu L\Big(\frac{\Es}{4\M\sigg}\Big)^2}_{\text{(c)}},\]\[
\underbrace{
\frac{\sqrt{\kappa}}{\kappa+1}\frac{\Es^2}{\sqrt{2}\M(4\M+1)\sigc\sqrt{\DNR}}
}_{\text{(d)}},
\underbrace{
\frac{\Es}{(4\M+1)\sqrt{\DNR}}
}_{\text{(e)}},
\underbrace{\vphantom{\frac{\Es}{(4\M+1)\sqrt{\DNR}}}1}_{\text{(f)}}
\Bigg\}.
\]
Note that
$\text{(b)}=\sqrt{\text{(a)}{\cdot}\text{(c)}}$
and $\text{(d)}=\sqrt{\text{(a)}{\cdot}\text{(f)}}$,
hence
$\text{(b)}\ge \min\{\text{(a)},\text{(c)}\}$
and $\text{(d)}\ge \min\{\text{(a)},\text{(f)}\}$.
Therefore, the terms (b) and (d) in the expression of $\etas$ can never be the smallest among the arguments of the minimum,
and can thus  be discarded. This yields \eqref{etaopt} and proves the lemma.
\end{proof}
\vspace{-5mm}
\section{Proof of Lemma \ref{lembudget}}
\label{proofoflembudget}

\begin{proof}
If $\Es\leq(3\M+1)\Eloc$, Local SGD is unfeasible, hence Full DGD is the only feasible option.
Otherwise ($\Es{>}(3\M+1)\Eloc$), the stage-wise budget costs
 under Local SGD or Full DGD are given by \eqref{budgbo} and \eqref{budgbo2}:
\[
\Bg
{=}
\bg\,\left\lceil
\frac{\ln(3\M)}{-\ln(1-\mu\etag)}
\right\rceil,\ \Bd
{=}
\bd\,\left\lceil
\frac{\ln(4\M)}{-\ln(1-\mu\etad)}
\right\rceil,
\]
where $\etag$ and $\etad$ are the optimal stepsizes, given in \eqref{optetalocsgd}
and \eqref{etaopt}, respectively.
Since $\ln(3\M)<\ln(4\M)$ and $\bg\leq\bd$, it follows
\[
\Bg
<
\bd\,\left\lceil
\frac{\ln(4\M)}{-\ln(1-\mu\etag)}
\right\rceil.
\]
Therefore, to prove that $\Bg\leq\Bd$, it suffices to show that
\(
\etad\leq\etag.
\)
Using the expression of $\etag$ in \eqref{optetalocsgd}, this condition is equivalent to
\[
\etad\leq\frac{2}{L+\mu},\ 
\etad\leq\frac{2\mu L}{L+\mu}\left(\frac{\Es}{3\M\sigg}\right)^2.
\]
The first condition holds trivially since $\etad\leq\frac{1}{L+\mu}$ (term (f) in \eqref{etaopt}).
The second one holds trivially since $\etad\leq\frac{2\mu L}{L+\mu}\Big(\frac{\Es}{4\M\sigg}\Big)^2$ (term (c) in \eqref{etaopt}), thus proving the first part of the lemma.

Upon entering Full DGD, 
the error satisfies $\Es\le (3\M+1)\Eloc$.
At the same time,
from \eqref{Einit.c}, operating  in the initialization-dominated regime of Full DGD requires
\(
\Es
\ge
(4\M+1)\sqrt{\DNR}.
\)
A necessary condition is
\(
(3\M+1)\Eloc
\ge
(4\M+1)\sqrt{\DNR}.
\)
Equivalently, 
using \eqref{DNR},
\[
\frac{3\M+1}{4\M+1}
\ge
\frac{2\kappa}{L+\mu}
\frac{\|\nabla f(\xstar)\|}
{(1-\lambda_2)\Eloc}.
\]
However,
$$
\frac{4}{5}
\overset{(a)}{>}
\frac{3\M+1}{4\M+1}
\text{ and }\frac{2\kappa}{L+\mu}\frac{\|\nabla f(\xstar)\|}{(1-\lambda_2)\Eloc}
\twocol{$$$$}
\overset{(b)}{\geq}\frac{4\kappa}{(\kappa+1)(1-\lambda_2)}
\overset{(c)}{\geq}\frac{2}{1-\lambda_2}
\overset{(d)}{\geq}1,
$$
where (a) uses $\M>1$, (b) uses
 \eqref{objheter}, 
(c) uses
$\kappa\geq 1$,
and (d) uses $\lambda_2\geq-1$.
Therefore, $(3\M+1)\Eloc< (4\M+1)\sqrt{\DNR}$, and the initialization-dominated regime of Full DGD is infeasible.
\end{proof}
\vspace{-5mm}
\section{Multi-Stage Budget Bounds}
\label{app:cumbud}

\begin{theorem}
\label{thm:cumbud}
Let $E_{k_s}\le \Es$ be the RMSE error bound at the start of stage $s$, and let $\errt<\Es$ be a target error level. Then, $\E{s+S}\le \errt$ after
\begin{align}
\label{Sdef}
S \triangleq \left\lceil \log_\M\left(\Es/\errt\right) \right\rceil
\end{align}
stages.
Furthermore, 
assume that the stage lengths satisfy
\begin{align}
\label{fghjfsd}
\delta_r \le
\lceil
\nu\;\E{r}^{\,-n}
\rceil,
\qquad r=s,\ldots,s+S-1,
\end{align}
for some constants \(\nu>0\) and \(n\ge0\).
Then,
if each iteration incurs a cost of \(b\),
 the total budget required to reach $E_{k_{s+S}}{\le}\errt$ satisfies
\begin{align}
b\cdot
\sum_{r=s}^{s+S-1} \delta_r
\le
b\cdot\begin{cases}
S+\frac{\nu\M^{n}}{\M^{n}-1}\errt^{-n}, & n>0,\\
S+\frac{\nu}{\ln(\M)}\ln\left(\frac{\M\Es}{\errt}\right), & n=0.
\end{cases}
\label{Kbound}
\end{align}
Furthermore, suppose that $\Es^{\,n}\le 2\nu$ and that there exists $\alpha\ge0$ such that,
$\forall r=s,\dots,s+S-1$,
\begin{align}
\label{alpha}
E_t \le \alpha\,\E{r},
\qquad
\forall t=k_r,\dots,k_{r+1}-1.
\end{align}
Then, for all $t=k_s,\dots,k_{s+S}-1$,
\begin{align}
\label{Ebound}
E_{t}
\leq 
\begin{cases}
\frac{\M\alpha\left(\frac{2\nu}{\M^n-1}\right)^{1/n}}{(t-k_s+1)^{1/n}},
 & n>0 \\
\alpha\Es\cdot e^{-\frac{\ln(\M)}{2\nu}(t-k_s+1)}
, & n=0.
\end{cases}
\end{align}
\end{theorem}

\begin{remark}
The regimes investigated in \secref{locsgdan} and \secref{dgdan} satisfy the conditions of Theorem~\ref{thm:cumbud} for specific choices of $\alpha$, $\errt$, $\nu$, and $n$.
\end{remark}

\begin{proof}
Consider a generic stage $s$ starting at iteration $k_s$, with error bound $\Es$. 
We aim to bound the total number of stages $S$ required to ensure $\E{s+S}\leq \errt$, the corresponding number of iterations $K=k_{s+S}-k_s$ and budget cost $b\cdot K$.
Under the conditions of the theorem, the error is reduced by a factor $\M$ at every stage. Hence $S$ is the smallest integer such that
$
\E{s+S}=\Es\cdot \M^{-S}\leq \errt,
$
which yields \eqref{Sdef}.

We now bound the corresponding number of iterations. 
In stage $r=s,\dots, s+S-1$, the stage length is bounded as
 in \eqref{fghjfsd}. Using $\lceil x\rceil\leq 1+x$ and $\E{r}=\M^{-(r-s)}\Es$ yields
$$
\delta_{r}\leq 1+
\nu \M^{(r-s) n}\Es^{-n}.
$$
If $n>0$, we then bound the total number of iterations as
$$
K=\sum_{r=s}^{s+S-1}\delta_{r}\leq
S+\nu\sum_{q=0}^{S-1}\M^{q n}\Es^{-n}
=
S+\frac{\nu}{\Es^n}\frac{\M^{nS}-1}{\M^{n}-1}
\twocol{$$$$}
\leq
S+\frac{\nu}{\M^{n}-1}\Big(\frac{\M^S}{\Es}\Big)^n.
$$
Finally, note that by definition of $S$
$$
\frac{\Es}{\M^{S-1}}>\errt
\quad\text{i.e.,}\quad
\frac{\M^{S}}{\Es}< \frac{\M}{\errt}.
$$
Using this bound, we then  obtain \eqref{Kbound} for the case $n>0$.

By repeating the same steps for the special case $n=0$,
we find $\delta_r\leq 1+\nu$, hence
$
K\leq S+S\nu$.
Using
$\frac{\M^{S}}{\Es}< \M/\errt\Rightarrow S\leq \log_\M(\M\Es/\errt)$,
yields \eqref{Kbound} for the case $n=0$.

Finally, consider $t = k_s,\dots,k_s+K-1$ and let $r$ be the unique integer such that
\begin{align}
\label{ghjkgfd}
\sum_{q=s}^{r-1}\delta_{q}\leq t-k_s\leq\sum_{q=s}^{r}\delta_{q}-1.
\end{align}
Then the $t$-th iteration falls in stage $r$, and therefore
\begin{align}
\label{ghjkgfd2}
E_{t}\le \alpha\,\E{r}= \alpha\Es\,\M^{-(r-s)}
\end{align}
by the condition \eqref{alpha} of the theorem.
To further bound $\M^{-(r-s)}$, note that
\begin{align}
\delta_q \le
\left\lceil
\frac{\nu}{\E{q}^{\,n}}
\right\rceil
=
\left\lceil
\frac{\nu\M^{(q-s)n}}{\Es^{\,n}}
\right\rceil
\leq \frac{2\nu\M^{(q-s)n}}{\Es^{\,n}},
\end{align}
where we used the fact that $\frac{\nu\M^{(r-s)n}}{\Es^{\,n}}\geq\frac{\nu}{\Es^{\,n}}\geq 0.5$
under the condition $\Es^n\leq2\nu$ of the theorem, 
combined with $\lceil x\rceil\leq 2x$ for $x\geq 0.5$.
Continuing from \eqref{ghjkgfd} for the case $n>0$, we obtain
\[t-k_s+1\leq
\sum_{q=s}^{r}\delta_{q}
\leq
2\nu\sum_{q=s}^{r}\frac{\M^{(q-s) n}}{\Es^n}
\leq
\frac{2\nu}{\Es^n}\frac{\M^{n(r-s+1)}}{\M^n-1}.
\]
Equivalently,
\[
\M^{-(r-s)}\leq
\M\Big(\frac{(\M^n-1)\Es^n}{2\nu}(t-k_s+1)\Big)^{-1/n},
\]
so that \eqref{Ebound} 
follows
for the case $n>0$, after replacing this bound into \eqref{ghjkgfd2}.

Repeating the same steps
for $n=0$ gives
\[
t-k_s+1\leq 
\sum_{q=s}^{r}\delta_{q}\leq 2\nu(r-s+1).
\]
Equivalently,
\(
\M^{-(r-s)}\leq \M^{-\frac{t-k_s+1}{2\nu}},
\)
so that \eqref{Ebound} 
follows after replacing this bound into \eqref{ghjkgfd2}.
The theorem is thus proved.
\end{proof}
\vspace{-5mm}
   \bibliographystyle{IEEEtran}
   \twocol{\balance}
\bibliography{IEEEabrv,biblio}

\begin{thebibliography}{10}
\providecommand{\url}[1]{#1}
\csname url@samestyle\endcsname
\providecommand{\newblock}{\relax}
\providecommand{\bibinfo}[2]{#2}
\providecommand{\BIBentrySTDinterwordspacing}{\spaceskip=0pt\relax}
\providecommand{\BIBentryALTinterwordstretchfactor}{4}
\providecommand{\BIBentryALTinterwordspacing}{\spaceskip=\fontdimen2\font plus
\BIBentryALTinterwordstretchfactor\fontdimen3\font minus
  \fontdimen4\font\relax}
\providecommand{\BIBforeignlanguage}[2]{{%
\expandafter\ifx\csname l@#1\endcsname\relax
\typeout{** WARNING: IEEEtran.bst: No hyphenation pattern has been}%
\typeout{** loaded for the language `#1'. Using the pattern for}%
\typeout{** the default language instead.}%
\else
\language=\csname l@#1\endcsname
\fi
#2}}
\providecommand{\BIBdecl}{\relax}
\BIBdecl

\bibitem{6494683}
S.~Kar and J.~M. Moura, ``{Consensus + innovations distributed inference over
  networks: cooperation and sensing in networked systems},'' \emph{IEEE Signal
  Processing Magazine}, vol.~30, no.~3, pp. 99--109, 2013.

\bibitem{Nedic2018}
A.~Nedi\'{c}, J.-S. Pang, G.~Scutari, and Y.~Sun, \emph{Multi-agent
  Optimization}, 1st~ed.\hskip 1em plus 0.5em minus 0.4em\relax Springer, Cham,
  2018.

\bibitem{YANG2019278}
T.~Yang, X.~Yi, J.~Wu, Y.~Yuan, D.~Wu, Z.~Meng, Y.~Hong, H.~Wang, Z.~Lin, and
  K.~H. Johansson, ``{A survey of distributed optimization},'' \emph{Annual
  Reviews in Control}, vol.~47, pp. 278--305, 2019.

\bibitem{10103556}
Y.~Ji, G.~Scutari, Y.~Sun, and H.~Honnappa, ``{Distributed (ATC) Gradient
  Descent for High Dimension Sparse Regression},'' \emph{IEEE Transactions on
  Information Theory}, vol.~69, no.~8, pp. 5253--5276, 2023.

\bibitem{9475989}
Y.~Xiao, Y.~Ye, S.~Huang, L.~Hao, Z.~Ma, M.~Xiao, S.~Mumtaz, and O.~A. Dobre,
  ``{Fully Decentralized Federated Learning-Based On-Board Mission for UAV
  Swarm System},'' \emph{IEEE Communications Letters}, vol.~25, no.~10, pp.
  3296--3300, 2021.

\bibitem{8950073}
S.~Savazzi, M.~Nicoli, and V.~Rampa, ``{Federated Learning With Cooperating
  Devices: A Consensus Approach for Massive IoT Networks},'' \emph{IEEE
  Internet of Things Journal}, vol.~7, no.~5, pp. 4641--4654, 2020.

\bibitem{Nedic2009grad}
A.~Nedi\'{c} and A.~Ozdaglar, ``Distributed subgradient methods for multi-agent
  optimization,'' \emph{{IEEE} Trans. Autom. Control}, vol.~54, no.~1, pp.
  48--61, Jan. 2009.

\bibitem{Yuan2016}
K.~Yuan, Q.~Ling, and W.~Yin, ``{On the Convergence of Decentralized Gradient
  Descent},'' \emph{SIAM Journal on Optimization}, vol.~26, no.~3, pp.
  1835--1854, 2016.

\bibitem{Nedic2009}
A.~Nedi\'{c}, A.~Olshevsky, A.~Ozdaglar, and J.~N. Tsitsiklis, ``On distributed
  averaging algorithms and quantization effects,'' \emph{{IEEE} Trans. Autom.
  Control}, vol.~54, no.~11, pp. 2506--2517, Nov. 2009.

\bibitem{Lian17}
X.~Lian, C.~Zhang, H.~Zhang, C.-J. Hsieh, W.~Zhang, and J.~Liu., ``Can
  decentralized algorithms outperform centralized algorithms? {A} case study
  for decentralized parallel stochastic gradient descent,'' in \emph{Proc. 31st
  NeurIPS}, Dec. 2017.

\bibitem{Kar2009}
S.~Kar and J.~M.~F. Moura, ``Distributed consensus algorithms in sensor
  networks with imperfect communication: Link failures and channel noise,''
  \emph{{IEEE} Trans. Inf. Theory}, vol.~57, no.~1, pp. 355--369, Jan. 2009.

\bibitem{10947567}
E.~G. Larsson and N.~Michelusi, ``Unified analysis of decentralized gradient
  descent: A contraction mapping framework,'' \emph{IEEE Open Journal of Signal
  Processing}, vol.~6, pp. 507--529, 2025.

\bibitem{10680589}
N.~Michelusi, ``Non-coherent over-the-air decentralized gradient descent,''
  \emph{IEEE Transactions on Signal Processing}, vol.~72, pp. 4618--4634, 2024.

\bibitem{michelusi26}
\BIBentryALTinterwordspacing
------, ``Interference-robust non-coherent over-the-air computation for
  decentralized optimization,'' 2026, {IEEE ICC, to appear}. [Online].
  Available: \url{https://arxiv.org/abs/2602.12426}
\BIBentrySTDinterwordspacing

\bibitem{9562482}
R.~Saha, S.~Rini, M.~Rao, and A.~J. Goldsmith, ``Decentralized optimization
  over noisy, rate-constrained networks: Achieving consensus by communicating
  differences,'' \emph{IEEE Journal on Selected Areas in Communications},
  vol.~40, no.~2, pp. 449--467, 2022.

\bibitem{9772390}
Z.~Jiang, G.~Yu, Y.~Cai, and Y.~Jiang, ``{Decentralized Edge Learning via
  Unreliable Device-to-Device Communications},'' \emph{IEEE Transactions on
  Wireless Communications}, vol.~21, no.~11, pp. 9041--9055, 2022.

\bibitem{9716792}
H.~Ye, L.~Liang, and G.~Y. Li, ``{Decentralized Federated Learning With
  Unreliable Communications},'' \emph{IEEE Journal of Selected Topics in Signal
  Processing}, vol.~16, no.~3, pp. 487--500, 2022.

\bibitem{9838891}
E.~Jeong, M.~Zecchin, and M.~Kountouris, ``{Asynchronous Decentralized Learning
  over Unreliable Wireless Networks},'' in \emph{IEEE International Conference
  on Communications}, 2022, pp. 607--612.

\bibitem{8786146}
A.~Reisizadeh, A.~Mokhtari, H.~Hassani, and R.~Pedarsani, ``An exact quantized
  decentralized gradient descent algorithm,'' \emph{IEEE Transactions on Signal
  Processing}, vol.~67, no.~19, pp. 4934--4947, 2019.

\bibitem{Koloskova2019}
A.~Koloskova, S.~U. Stich, and M.~Jaggi, ``Decentralized stochastic
  optimization and gossip algorithms with compressed communication,'' in
  \emph{Proc. 36th ICML}, Jun. 2019.

\bibitem{9782148}
N.~Michelusi, G.~Scutari, and C.-S. Lee, ``{Finite-Bit Quantization for
  Distributed Algorithms With Linear Convergence},'' \emph{IEEE Transactions on
  Information Theory}, vol.~68, no.~11, pp. 7254--7280, 2022.

\bibitem{Shi2015EXTRA}
W.~Shi, Q.~Ling, G.~Wu, and W.~Yin, ``{EXTRA: An Exact First-Order Algorithm
  for Decentralized Consensus Optimization},'' \emph{SIAM J. Optim.}, vol.~25,
  pp. 944--966, May 2015.

\bibitem{8491372}
K.~Yuan, B.~Ying, X.~Zhao, and A.~H. Sayed, ``{Exact Diffusion for Distributed
  Optimization and Learning--Part I: Algorithm Development},'' \emph{IEEE
  Transactions on Signal Processing}, vol.~67, no.~3, pp. 708--723, 2019.

\bibitem{7798263}
G.~Qu and N.~Li, ``Harnessing smoothness to accelerate distributed
  optimization,'' in \emph{2016 IEEE 55th Conference on Decision and Control
  (CDC)}, 2016, pp. 159--166.

\bibitem{7398129}
P.~D. Lorenzo and G.~Scutari, ``Next: In-network nonconvex optimization,''
  \emph{IEEE Transactions on Signal and Information Processing over Networks},
  vol.~2, no.~2, pp. 120--136, 2016.

\bibitem{Tang2018}
H.~Tang, S.~Gan, C.~Zhang, T.~Zhang, and J.~Liu, ``Communication compression
  for decentralized training,'' in \emph{Proc. 32nd NeurIPS}, Dec. 2018.

\bibitem{7405263}
A.~Nedi\'{c} and A.~Olshevsky, ``{Stochastic Gradient-Push for Strongly Convex
  Functions on Time-Varying Directed Graphs},'' \emph{IEEE Transactions on
  Automatic Control}, vol.~61, no.~12, pp. 3936--3947, 2016.

\bibitem{9563232}
H.~Xing, O.~Simeone, and S.~Bi, ``{Federated Learning Over Wireless
  Device-to-Device Networks: Algorithms and Convergence Analysis},'' \emph{IEEE
  Journal on Selected Areas in Communications}, vol.~39, no.~12, pp.
  3723--3741, 2021.

\bibitem{Kovalev2020}
D.~Kovalev, A.~Koloskova, M.~Jaggi, P.~Richt\'{a}rik, and S.~U. Stich, ``A
  linearly convergent algorithm for decentralized optimization: Sending less
  bits for free!'' in \emph{Proc. 24th AISTATS}, Apr. 2021.

\bibitem{Liao2021}
Y.~Liao, Z.~Li, K.~Huang, and S.~Pu, ``A compressed gradient tracking method
  for decentralized optimization with linear convergence,'' \emph{IEEE Trans.
  on Automatic Control}, vol.~67, no.~10, pp. 5622--5629, 2022.

\bibitem{Magnusson2020}
S.~{Magn\'usson}, H.~{Shokri-Ghadikolaei}, and N.~{Li}, ``{On Maintaining
  Linear Convergence of Distributed Learning and Optimization Under Limited
  Communication},'' \emph{{IEEE} Trans. Signal Process.}, vol.~68, pp.
  6101--6116, 2020.

\bibitem{berahas2018balancing}
A.~S. Berahas, R.~Bollapragada, N.~S. Keskar, and E.~Wei, ``{Balancing
  Communication and Computation in Distributed Optimization},'' \emph{IEEE
  Transactions on Automatic Control}, vol.~64, no.~8, pp. 3141--3155, 2019.

\bibitem{berahas2021convergence}
A.~S. Berahas, R.~Bollapragada, and E.~Wei, ``{On the Convergence of Nested
  Decentralized Gradient Methods With Multiple Consensus and Gradient Steps},''
  \emph{IEEE Transactions on Signal Processing}, vol.~69, pp. 4192--4203, 2021.

\bibitem{choi2022convergence}
W.~Choi, D.~Kim, and S.-B. Yun, ``{Convergence results of a nested
  decentralized gradient method for non-strongly convex problems},''
  \emph{Journal of Optimization Theory and Applications}, vol. 195, no.~1, pp.
  172--204, Aug. 2022.

\bibitem{6854643}
V.~Schwarz, G.~Hannak, and G.~Matz, ``{On the convergence of average consensus
  with generalized metropolis-hasting weights},'' in \emph{IEEE International
  Conference on Acoustics, Speech and Signal Processing (ICASSP)}, 2014, pp.
  5442--5446.

\bibitem{7721743}
M.~Centenaro, L.~Vangelista, A.~Zanella, and M.~Zorzi, ``Long-range
  communications in unlicensed bands: the rising stars in the iot and smart
  city scenarios,'' \emph{IEEE Wireless Communications}, vol.~23, no.~5, pp.
  60--67, 2016.

\bibitem{matrixanalysis}
R.~A. Horn and C.~R. Johnson, \emph{Matrix Analysis}.\hskip 1em plus 0.5em
  minus 0.4em\relax Cambridge University Press, 1999.

\bibitem{doi:10.1137/070704277}
A.~Nemirovski, A.~Juditsky, G.~Lan, and A.~Shapiro, ``{Robust Stochastic
  Approximation Approach to Stochastic Programming},'' \emph{SIAM Journal on
  Optimization}, vol.~19, no.~4, pp. 1574--1609, 2009.

\bibitem{6259916}
F.~Iutzeler, P.~Ciblat, and J.~Jakubowicz, ``Analysis of max-consensus
  algorithms in wireless channels,'' \emph{IEEE Transactions on Signal
  Processing}, vol.~60, no.~11, pp. 6103--6107, 2012.

\bibitem{florescu2013handbook}
I.~Florescu and C.~A. Tudor, \emph{Handbook of Probability}.\hskip 1em plus
  0.5em minus 0.4em\relax Wiley, 2013.

\bibitem{rudin1953principles}
W.~Rudin, \emph{Principles of Mathematical Analysis}.\hskip 1em plus 0.5em
  minus 0.4em\relax McGraw-Hill, 1953.

\end{thebibliography}

\end{document}